\pdfoutput=1
\documentclass[11pt,british,reqno]{article}
\usepackage{lmodern}
\usepackage{avant}
\usepackage{courier}

\usepackage[T1]{fontenc}
\usepackage[latin9]{inputenc}
\usepackage[a4paper]{geometry}
\usepackage{color}
\usepackage{pifont}
\usepackage{float,subfig}
\usepackage{amsmath,amsthm,amssymb}
\usepackage{stmaryrd}
\usepackage{graphicx}
\usepackage{setspace}
\usepackage[titletoc,title]{appendix}
\usepackage[all]{xy}
\usepackage{mathrsfs}
\usepackage[numbers]{natbib}
\usepackage[unicode=true,pdfusetitle,
 bookmarks=true,bookmarksnumbered=false,bookmarksopen=false,
 breaklinks=false,pdfborder={0 0 0},pdfborderstyle={},backref=false,colorlinks=true]
 {hyperref}
\hypersetup{
 linkcolor=magenta, urlcolor=blue, citecolor=blue}

\makeatletter

\numberwithin{equation}{section}
\numberwithin{figure}{section}
\numberwithin{table}{section}
 \theoremstyle{definition}
  
  \theoremstyle{definition}
  
  \theoremstyle{plain}
  \newtheorem{lemma}{\protect\lemmaname}
  \theoremstyle{plain}
  \newtheorem{proposition}{\protect\propositionname}
  \theoremstyle{remark}
  
  \theoremstyle{plain}
  \newtheorem{corollary}{\protect\corollaryname}

\SetSymbolFont{stmry}{bold}{U}{stmry}{m}{n}

\numberwithin{equation}{section}
\numberwithin{figure}{section}
\numberwithin{table}{section}
\theoremstyle{plain}

\makeatletter
\renewenvironment{proof}[1][\proofname]{\par
  \pushQED{\qed}%
  \normalfont \topsep6\p@\@plus6\p@\relax
  \list{}{%
    \settowidth{\leftmargin}{\itshape\proofname:\hskip\labelsep}%
    \setlength{\labelwidth}{0pt}%
    \setlength{\itemindent}{-\leftmargin}%
  }%
  \item[\hskip\labelsep\itshape#1\@addpunct{:}]\ignorespaces
}{%
  \popQED\endlist\@endpefalse
}
\makeatother

\date{} 

\makeatother

  \providecommand{\definitionname}{Definition}
  \providecommand{\examplename}{Example}
  \providecommand{\lemmaname}{Lemma}
  \providecommand{\propositionname}{Proposition}
  \providecommand{\remarkname}{Remark}
\providecommand{\corollaryname}{Corollary}

\begin{document}

\title{Tropical limit and a micro-macro correspondence in statistical physics}

\author{M. Angelelli}
\maketitle
\begin{center}
Department of Mathematics and Physics 
\par\end{center}

\begin{center}
``Ennio De Giorgi'', University of Salento and sezione INFN, 
\par\end{center}

\begin{center}
Lecce 73100, Italy 
\par\end{center}

\begin{abstract}

Tropical mathematics is used to establish a correspondence between certain microscopic and macroscopic objects in statistical models. Tropical algebra gives a common framework for macrosystems (subsets) and their elementary
constituent (elements) that is well-behaved with respect to composition. This kind of connection is studied with maps that preserve a monoid structure. The approach highlights an underlying order relation that is explored through the concepts of filter and ideal. Main attention is paid to asymmetry and duality between $\max$- and $\min$-criteria. Physical implementations are presented through simple examples in thermodynamics and non-equilibrium physics. The phenomenon of ultrametricity, the notion of tropical equilibrium and the role of ground energy in non-equilibrium models are discussed. Tropical symmetry, i.e. idempotence, is investigated.  

\end{abstract}

\vspace{2pc}


\section{Introduction}

The distinction between macroscopic and microscopic representations of phenomena is one of the fundamental problems in physics. This subject has generated many profound questions and techniques whose relevance goes beyond statistical physics. For example, the transition between the molecular dynamics and the human (thermodynamic) length scales is still investigated \cite{Bartlett1980,Posch1997,Crooks1999} and gave rise to the notion of statistical entropy, which has proved to be a fundamental tool in many areas of modern sciences \cite{Gray2011,Holzinger2014,Kahraman2016}. 

More broadly, the ``micro/macro'' paradigm involves many situations where different descriptions of a system and their relative complexity have concrete effects. This is also a practical issue, since complex systems are now pervasive in many branches of science \cite{Nicolis2012} and a deeper understanding of (dis-)similarities between elementary and emergent phenomena is a key point. In the cases when a collective behaviour is not reducible to its individual constituents, one can recognize complexity in the composition of the elementary entities. Hence, it is possible that a change in the composition rules affects the relative complexity. This approach can be used to highlight analogies between the micro- and the macro-sectors, rather than their differences. 

In the present work, we follow this path and concentrate on a correspondence between micro- and macro-physics starting from associative rules. In particular, we argue that tropical algebra \cite{MS2015} provides one with a common framework to deal with both these descriptions. Tropical limit is usually derived from the real or complex setting by means of the change of variable
\begin{equation}
|x|\mapsto\exp\left(\frac{X}{\varepsilon}\right)\label{eq: tropical change variable}
\end{equation}
that induces a tropical algebra in the limit $\varepsilon\rightarrow0$
\begin{equation}
X\oplus Y:=\lim_{\varepsilon\rightarrow0^{+}}\varepsilon\cdot\ln\left(e^{\frac{X}{\varepsilon}}+e^{\frac{Y}{\varepsilon}}\right).\label{eq: standard tropical limit sum}
\end{equation}
Elementary entities $X$ are combined through $\oplus$, which is a ``shadow'' of the usual composition for variables ${\displaystyle \exp\left(\frac{X}{\varepsilon}\right)}$ with exponential complexity in $X$. This is a hint on advantages of tropical limit: it forgets part of the information in a system in order to highlight an underlying structure, which is often more practical to manage and still non-trivial. 

Tropical limit in statistical physics was discussed in \cite{AK2015}. There, it has been argued that the Boltzmann constant
$k_{B}$ is a proper parameter to highlight a combinatorial skeleton
of some statistical models. In several ordinary cases, the tropical
and low-temperature limits coincide \cite{Sturmfels2004,Kapranov2011,Marcolli2014}.
More generally, phenomena such as exponential degenerations
\cite{Pauling1935,Diep2005,GutzowSchemler2009} and negative
and limiting temperatures \cite{Ramsey1956,Rumer1960} can be easily included in the limit $k_{B}\rightarrow 0$ and make it non-trivial. 
The analysis suggests that the tropical limit of such statistical models preserves
an associated relational structure, namely an order relation. Algebraic
methods for ordered sets and logic have recently been developed, see
e.g. \cite{GuidoToto2008}, and also tropical mathematics has benefited
from these tools \cite{LS2014}. 

Our aim is to deepen these concepts in order to establish
a correspondence between certain micro- and macro-systems in accordance with
tropical composition, hence the order structure. More specifically, we look at monoids, that are defined by a set $\Lambda$ and a composition $\oplus$ (associative operation) on $\Lambda$ with a distinguished element $\infty$ that is neutral for $\oplus$. The assumption of this simple algebraic structure emphasises the role of $\infty$, which is an extremum for the associated order relation. Cases when it is the only extremum are relevant in terms of symmetry breaking between dual orders. As we will see, this last point draws attention to the issue of the ground energy in general statistical models. 

For every tropical structure (\emph{idempotent monoid}) $(\Lambda,\oplus,\infty)$, one also gets
another tropical structure, that is the set $\mathcal{P}(\Lambda)$
of the subsets of $\Lambda$. The tropical operation on $\mathcal{P}(\Lambda)$
is the set-theoretic join $\cup$ (respectively, intersection $\cap$)
and the neutral element is $\emptyset$ (respectively, $\Lambda$).
If the elements of $\Lambda$ represent physical microscopic systems, subsets
of $\Lambda$ are ``macrosystems''. In this perspective,
we address the question of structure-preserving connections (i.e., monoid
homomorphisms) between a tropical system $\Lambda$ and the associated
power set $\mathcal{P}(\Lambda)$. 

For this purpose, the concept of \emph{filter} will be pivotal. Filters,
especially ultrafilters, play a fundamental part in mathematical logic
\cite{Mendelson2015}. Since basic rules of classical logic are
strictly related to probability axioms \cite{Jaynes2003}, it is
not surprising that a different way to deal with probabilities needs
a different logic. In our physical framework, this difference
between classical and ``tropical'' logic can be simply expressed
as a broken symmetry between disjunctions (joins, $\sup$, existential
quantifier $\exists$) and conjunctions (intersections, $\inf$, universal
quantifier $\forall$). This symmetry is characteristic of classical probability and Boolean algebras. For example, the principle of inclusion-exclusion 
\begin{equation}
p(A\cup B)=p(A)+p(B)-p(A\cap B)\label{eq: PIE}
\end{equation}
relates the probability of the disjunction of events $A$ and $B$
to the probabilities of the events themselves and of the conjunction.
The tropical approach leads one to get information on conjunction
or disjunction, depending on its presentation via the tropical sum $\oplus=\min$ or
$\oplus=\max$, but not both of them simultaneously. This broken
symmetry is algebraically translated in the absence of a subtraction and
this affects basic counting principles such as (\ref{eq: PIE}). On
the order-theoretical side, this process affects order-reversing dualities
(e.g. set complements) that preserve algebraic properties. In particular,
the dual of the neutral element is not an element of the algebra itself,
for instance $-\infty$ is an element of the $\max$-plus algebra $(\mathbb{R}\cup\{-\infty\},\max,-\infty)$ 
but $+\infty$ is not. 

This language can be effective in applications, with special attention to non-equilibrium \cite{Langer1969,Newman1980}, metastable \cite{Newman1997}, and disordered systems \cite{Derrida1983}. A prominent role is played by spin glasses \cite{MPV1987},
whose study has prompted many theoretical and computational techniques 
\cite{BinderYoung1986}, with applications in quantum field theories \cite{RyuTakayanagi2006}, message passing \cite{Mezard2014} and artificial intelligence \cite{Mezard2002,Mezard2005}. 
A major breakthrough in the study of spin glasses was done by Parisi
\cite{Parisi1980} with the Replica Symmetry Breaking ansatz. This
technique opened the way to the description of interesting phenomena
for a class of spin glasses. One of these is derived
from the overlap distribution among different replicas and is called
ultrametricity. An ultrametric $d$ on a certain space is a metric
which satisfies a stronger version of the triangle inequality 
\begin{equation}
d(\alpha,\beta)\leq\max\{d(\alpha,\gamma),d(\gamma,\beta)\}.\label{eq: ultrametric}
\end{equation}
The occurrence in (\ref{eq: ultrametric}) of a tropical addition
$\max$ instead of the standard real one $+$ and exponential degeneration
in spin glasses \cite{Sobotta1986} are hints of a link between this
kind of models and tropical algebraic structures.

We will show that filters give a rigorous but flexible background
to investigate tropical statistical systems and their physical aspects. First, ultrametricity comes into play, since filters and ultrametrics are in a relation
that resembles the one between topology and metric. Then, they allow to identify the type of objects involved in such a ``micro/macro'' correspondence. In addition, they provide one with a simple criterion to investigate non-equilibrium
and metastability, namely the invariance of the spectrum under translation that relates to different choices for the ground (zero point) energy. This issue is related to tropicalization processes of real variables. Filters also give a picture of the physical process to approach the tropical limit. Indeed, if the definitions in \cite{AK2015} are assumed, then tropical statistical
physics is derived from the double scaling limit for Boltzmann constant
$k_{B}\rightarrow0$ and Avogadro number $N_{A}\rightarrow\infty$
keeping fixed the product $R=k_{B}\cdot N_{A}$. Some consequences
of the limit for the reference cardinality $N_{A}$ in statistics
can be addressed with the filter language. For instance, this shows how enumeration
changes in the tropical limit. We will refer to this procedure as
a dequantification for probability weights. Other puzzling
concepts, such as the $n\rightarrow0$ limit for the dimension of
the replica overlap matrix in spin glass theory, could take a concrete shape
in this setting.

It should be remarked that we do not wish to deal with measure-theoretic
properties of specific models. Many interesting results in this direction
have been achieved during last decades, with particular regard to $p$-adic
models \cite{Grossman1989} and Ruelle probability cascade \cite{Ruelle1987}.
Here we focus on the unifying role that tropical limit plays in explaining
apparent contrasts between macro- and micro-physics. 

The paper is organized as follows. In Section \ref{sec: Notation} we remind few notions on monoids and order theory in order to make this paper self-contained. In Section \ref{sec: Tropical limit and the role of Boltzmann constant} we deal with some consequences of its physical identification with the limit $k_{B}\rightarrow0$. The discussion suggests what types of algebraic connections can be drawn in the tropical setting. This is formalized in Section \ref{sec: Filters and ultrametricity}, where filters and ideals are introduced and a simple relation with ultrametric is observed. In Section \ref{sec: Monoid homomorphisms and a set/element correspondence}
we investigate links between tropical structures via monoid homomorphisms and provide a characterization of linearly ordered tropical structure in this perspective. In Section \ref{sec: Perturbative tropical limit}
we encode the information on partition functions in a suitable
form for a ``perturbative'' study of tropical limit, since usual analytic expansions may fail in the tropical (non-analytic) limit. Subsequent sections apply these tools to concrete physical issues. In Section \ref{subsec: Duality and non-equilibrium} we use the filter language to introduce tropical equilibrium and discuss a basic non-equilibrium model with particular focus on the choice of the ground energy. This issue is deepened in Section \ref{subsec: Tropical probability and tropical symmetry} with the introduction of global
or local tropicalizations and the study of global or local tropical actions. In Section \ref{sec: Global and local tropical symmetry} the issue of tropicalizations of more variables, as a whole or one at a time, is addressed. Idempotence is achieved as a global or a local tropical symmetry. Its implications in counting processes lead to tropical probability, that is explored in terms of global (Subsection \ref{subsec: Global tropical symmetry and statistical amoebas}) or local (Subsection \ref{subsec: Local tropical symmetry and the dequantification procedure })
tropical symmetry. Finally, we draw conclusions and discuss future perspectives in Section \ref{sec: Conclusions}.  

\section{\label{sec: Notation} Notation and definitions}

Tropical geometry (see e.g. \cite{MS2015} for an introduction) is
a recent branch of algebraic geometry originating from earlier studies
in computer science, optimization and mathematical physics \cite{Cohen1999,Inoue2012,LS2014}.
It is based on usual algebraic concepts, like polynomials, ideals
and varieties, translated in the setting of an idempotent semiring.
This means that all algebraic expressions involve operations on a
semiring instead of a field $\mathbb{K}$ (usually
$\mathbb{R}$ or $\mathbb{C}$, or finite fields $\mathbb{F}_{q}$).
The basic example of a tropical structure is the $\max$-plus semiring $\mathbb{R}_{\max}$: it is defined as a $5$-uple $(\mathbb{R}\cup\{-\infty\},\oplus,\odot,-\infty,0)$
where $\mathbb{R}_{\max}:=\mathbb{R}\cup\{-\infty\}$, $a\oplus b:=\max\{a,b\}$
is the tropical addition, $a\odot b:=a+b$ is the tropical multiplication,
 $-\infty$ and $0$ are the neutral elements of
$\oplus$ and $\odot$ respectively. A common process to derive the $\max$-plus semiring from $\mathbb{R}$ or $\mathbb{C}$ comes through the limit
$\varepsilon\rightarrow0^{+}$ for the following family of ring operations
\begin{eqnarray}
x\oplus_{\varepsilon}y & := & \varepsilon\cdot\ln\left(e^{\frac{x}{\varepsilon}}+e^{\frac{x}{\varepsilon}}\right),\nonumber \\
x\odot_{\varepsilon}y & := &\varepsilon\cdot\ln\left(e^{\frac{x}{\varepsilon}}\cdot e^{\frac{x}{\varepsilon}}\right)=x+y.
\label{eq: tropical limit}
\end{eqnarray}
The peculiarity of the $\max$-plus semiring is that $\oplus$ is
idempotent, i.e. $x\oplus x=x$ for all $x$ in $\mathbb{R}_{\max}$.
So there exists a multiplicative inverse for each element in $\mathbb{R}=\mathbb{R}_{\max}\backslash\{-\infty\}$,
but there is no additive inverse for elements in $\mathbb{R}$. In
fact, $x\oplus y=-\infty$ implies $-\infty=x\oplus y=(x\oplus x)\oplus y=x\oplus(x\oplus y)=x\oplus(-\infty)=x$.
This means that there is no subtraction. Consequently, the extension of several useful mathematical tools, e.g. differentiation, to the tropical setting requires particular attention. This purpose has generated many additional techniques, such as ultradiscrete methods (see e.g. \cite{Grammaticos1997,Inoue2012}).  

An action of $\mathbb{R}_{\max}$ on a statistical model has relevant effects. In particular, the inclusion-exclusion formula (\ref{eq: PIE}) should be reconsidered since it involves subtraction. Even the equivalent expression $p(A\cup B)+p(A\cap B)=p(A)\cup p(B)$ in not informative in the tropical language, since $p(A\cap B)\leq p(A\cup B)$,
so $p(A\cup B)\oplus p(A\cap B)=p(A\cup B)$. 

For our purposes, it is worth looking at a more general \textit{tropical semiring}, that is a $5$-uple $(\Lambda,\oplus,\odot,\infty,0)$ where $\Lambda$ is a set, the addition $\oplus$ and multiplication $\odot$ are associative and commutative binary operations
on $\Lambda$ with identities and the distributivity property holds. $\infty$ is the neutral element for $\oplus$ and $0$ is the neutral element for $\odot$. All the elements of $\Lambda\backslash\{\infty\}$ are multiplicatively
invertible and $\oplus$ is idempotent, i.e. $a\oplus a=a$ for all $a\in\Lambda$. 

The \textit{tropical monoid} $(\Lambda,\oplus,\infty)$ is a sub-structure associated to a tropical semiring that forgets $\odot$ and $0$. A \textit{monoid
homomorphism} is a map $\psi:\,(\Lambda_{1},\oplus_{1},\infty_{1})\longrightarrow(\Lambda_{2},\oplus_{2},\infty_{2})$
such that $\psi(\infty_{1})=\infty_{2}$ and $\psi(x\oplus_{1}y)=\psi(x)\oplus_{2}\psi(y)$
for all $x,y\in\Lambda_{1}$. An \textit{erasing element} of a tropical
monoid $\Lambda$ is an element $\top\in\Lambda$ such that $a\oplus\top=\top$
for all $a\in\Lambda$. If $\Lambda$ has no erasing element, then
we call it a \textit{grounded monoid}. If one also looks at the multiplication, an erasing element $\top$ satisfies $a\odot\top\oplus b=a\odot(\top\oplus a^{-1}\odot b)=a\odot\top$
for all $a,b\in\Lambda$. By uniqueness of the erasing element, one has $a\cdot\top=\top$.

When no ambiguity arises, $\Lambda$ will denote both a tropical algebra
and the underlying set with a slight abuse of notation. Well-known
examples of tropical monoids with $\oplus=\max$ are $\mathbb{N}$,
$\mathbb{Q}^{\vee}:=\mathbb{Q}\cup\{-\infty\}$ and similarly $\mathbb{R}^{\vee}$,
while $\mathbb{R}^{\wedge}:=\mathbb{R}\cup\{+\infty\}$ involves $\oplus=\min$. 

From idempotence and associativity, it follows that the tropical addition
$\oplus$ induces an order on $\Lambda$, that is 
\begin{equation}
x\preceq y\Leftrightarrow x\oplus y=y\label{eq: order induced by tropical addition}
\end{equation}
and $\infty$ is the minimum element of $\Lambda$ with respect
to this order. We adopt the following notation for ordered sets: a \textit{partially ordered set} (or \textit{poset}) is a pair $(\Lambda,\preceq)$
where $\preceq$ is a binary reflexive, antisymmetric and transitive
relation on $\Lambda$. A poset is \textit{totally ordered} if any two elements $x,y\in\Lambda$ are comparable, i.e. $x\preceq y$ or $y\preceq x$ for all $x,y\in\Lambda$. The \textit{minimum} (respectively, \textit{maximum}) of a poset,
if it exists, is the unique element $\bot$ (respectively, $\top$),
such that $\bot\preceq x$ (respectively, $x\preceq\top$) for all
$x\in\Lambda$. The \textit{infimum} of a subset $\lambda\subseteq\Lambda$, if it
exists, is the unique element $\inf(\lambda)\in\Lambda$ such that
i.) $\inf(\lambda)\preceq x$ for all $x\in\lambda$, and ii.) if
$y\in\Lambda$ is such that $y\preceq x$ holds for all $x\in\lambda$,
then $y\preceq\inf(\lambda)$. The \textit{supremum} of a subset $\lambda\subseteq\Lambda$
is the infimum of $\lambda$ in the poset $(\Lambda,\succeq)$ given by the inverse
relation of $\preceq$. A \textit{join-semilattice} (respectively, \textit{meet-semilattice}) is a poset $(\Lambda,\preceq)$ such that each pair of elements of
$\Lambda$, or equivalently each finite subset of $\Lambda$, has
a supremum (respectively, an infimum). A \textit{lattice} is simultaneously
a join- and a meet-semilattice. A \textit{complete lattice} is a lattice $(\Lambda,\preceq)$ where arbitrary (even infinite) suprema and infima exist. 

With previous definitions, we will call \textit{grounded poset} a
poset $(\Lambda,\preceq,\bot)$ with a minimum $\bot$ and without
maximum. Note that a grounded tropical monoid can not be a complete lattice.
Indeed, $\displaystyle \sup\Lambda$ does not exists, since it
would be an erasing element.

Now we can introduce filters as follows. Given a poset  $(\Lambda,\preceq)$,
a\textit{ filter} on $\Lambda$ is a collection $\mathcal{F}\subseteq\Lambda$
of elements of $\Lambda$, such that the following properties hold: 
\begin{enumerate}
\item $\mathcal{F}\neq\emptyset$; 
\item for all $x,y\in\mathcal{F}$, there exists $z\in\mathcal{F}$ such
that $z\preceq x$ and $z\preceq y$ ($\mathcal{F}$ is downward directed); 
\item for all $x\in\mathcal{F}$, if $x\preceq y$ then $y\in\mathcal{F}$
($\mathcal{F}$ is upward closed). 
\end{enumerate}
If the following additional property holds 
\begin{equation}
A\in\mathcal{F}\Leftrightarrow A^{c}\notin\mathcal{F},\quad A\in\mathcal{P}(\Omega),\label{eq: ultrafilter property}
\end{equation}
one talks of $\mathcal{F}$ as an \textit{ultrafilter} on $\Omega$.
It is equivalent to the property that, for every filter $\mathcal{V}$
such that $\mathcal{F\subseteq V}$, one has $\mathcal{F=V}$, e.g.
$\mathcal{F}$ is a maximal filter. A filter such that $\bot\notin\mathcal{F}$,
i.e. $\mathcal{F}\neq\Lambda$, is said \textit{proper}. Note that
if $(\Lambda,\preceq)$ is a lattice, then property 2. in the definition of filters can be restated using property 3. as $\inf\{x,y\}\in\mathcal{F}$
for all $x,y,\in\mathcal{F}$. A filter is called a \textit{principal
filter} if it is of the form 
\begin{equation}
\mathcal{F}_{x}:=\left\{ y\in\Lambda:\,x\preceq y\right\} \label{eq: principal filter}
\end{equation}
for some $x\in\Lambda$. If $\mathcal{B}\subseteq\Lambda$ verifies
properties 1. and 2. and $\bot\notin\mathcal{B}$, then $\mathcal{B}$
is a \textit{proper filter base}. One can extend a filter base $\mathcal{B}$
to $\mathcal{F}:=\left\{ C\in\Lambda:\,\exists B\in\mathcal{B},\,B\preceq C\right\} $,
which is a filter. 

The dual notion of a filter is an \textit{ideal} on $\Lambda$, that is a collection $\mathcal{I}\subseteq\Lambda$
of elements of $\Lambda$ such that: 
\begin{enumerate}
\item $\mathcal{I}\neq\emptyset$; 
\item for all $x,y\in\mathcal{I}$, there exists $z\in\mathcal{I}$ such
that $x\preceq z$ and $y\preceq z$ ($\mathcal{I}$ is upward directed); 
\item for all $x\in\mathcal{I}$, if $y\preceq x$ then $y\in\mathcal{I}$
($\mathcal{I}$ is downward closed). 
\end{enumerate}
An ideal of the form 
\begin{equation}
\mathcal{I}_{y}=\left\{ x\in\Lambda:\,x\preceq y\right\} .\label{eq: principal ideal}
\end{equation}
is called a \textit{principal ideal}. If $\mathcal{B}\subseteq\Lambda$
verifies the properties 1. and 2. and $\top\notin\mathcal{B}$, then $\mathcal{B}$
is a \textit{proper ideal base}. One can extend an ideal base $\mathcal{B}$
to $\mathcal{I}:=\left\{ C\in\Lambda:\,\exists B\in\mathcal{B},\,C\preceq B\right\} $,
which is an ideal. 

In particular, for any tropical semiring one has $a\oplus(a\oplus b)=(a\oplus a)\oplus b=a\oplus b$, then $a\preceq a\oplus b$. In the same way, $b\preceq a\oplus b$.
If $z\in\Lambda$ is such that $a\preceq z$ and $b\preceq z$ then
$a\oplus z=z$ and $b\oplus z=z$, thus $(a\oplus b)\oplus z=a\oplus(b\oplus z)=a\oplus z=z$. So $a\oplus b\preceq z$. This corresponds to the fact that 
\begin{equation}
a\oplus b=\sup\{a,b\}\label{eq: sum sup}
\end{equation}
and $\Lambda$ is a join-semilattice. 

\section{\label{sec: Tropical limit and the role of Boltzmann constant} Tropical
limit and the role of Boltzmann constant} 

The physical parameter which controls the occurrence of real or tropical
features is the Boltzmann constant \cite{AK2015}. More to the point, given a physical statistical system specified by a partition function and an associated free energy, its \emph{tropical limit} is defined as the simultaneous limit for $k_{B}$ and Avogadro number $N_{A}$ such that their product is kept constant, i.e. 

\begin{equation}
k_{B}\rightarrow0^{+}, \qquad
N_{A}\rightarrow\infty: \qquad 
k_{B}\cdot N_{A}=:R \mbox{ is constant}.
\label{eq: tropical double limit}
\end{equation}

A deeper understanding of effects of the limit $k_{B}\rightarrow0$ is necessary since $k_{B}$ is the fundamental unit that connects the microscopic and macroscopic
worlds. As a first check, it should be noted that this is physically reasonable since $k_{B}\ll 1$ and $N_{A}\gg 1$, while the universal gas constant $R\sim\mathcal{O}(1)$ is the order of unity, so it is a macroscopically distinguishable quantity.

One can also draw an analogy with the $\hbar\rightarrow0$ limit
when one has a quantum theory and a consistent procedure, driven by
$\hbar$, that shut off quantum effects. This goes beyond a formal analogy since the limit $k_{B}\rightarrow 0$ has been discussed in stochastic processes \cite{Lavenda1982}, viewed as the limit of vanishing white noise, and in the analysis of thermodynamic complementarity and fluctuation theory \cite{Schlogl1988, VelazquezAbad2012}. These approaches establish a stronger correspondence between the  semiclassical limit $\hbar\rightarrow 0$ in quantum mechanics and the $k_{B}\rightarrow0$ limit. However, the corresponding ``classical'' theory has been interpreted as a thermodynamic limit, where a large number of particles $N$ are involved. It should be stressed that our approach is conceptually different from
thermodynamic limit(s), which involves the approximation of random
variables with their averages. By contrast, tropical limit is intended to provide a setting where statistical properties can be studied exactly through certain algebraic rules. 

Different statistical models can have the same tropical limit. In this sense, the limit $k_B\rightarrow 0$ provides a tropical classification that leaves control parameters (e.g., temperature $T$) as free variables. At the same time, the action on the reference cardinality $N_A$ in (\ref{eq: tropical double limit}) results in a different counting process and this affects statistics of a generic system through arithmetic. The scaling of $N_A$, as well as the algebraic structure that arises in the process, is not explicit in the deterministic theory obtained in \cite{Schlogl1988, VelazquezAbad2012}, even if some similarities can be drawn. The implications of the tropical limit in arithmetic and its connections with the thermodynamic formalism have been discussed in a different framework, namely number theory, see e.g. \cite{Marcolli2014} and references therein. However, these works introduce the tropical limit as a low temperature limit, i.e. $T\rightarrow 0$. There are both conceptual and practical differences with our approach. First, the limit $T\rightarrow 0 $ reduces the model to a part of the boundary of the space of thermodynamic variables given by the evaluation at $T=0$. As already remarked, we preserve the set of free parameters and, hence, thermodynamic relations, through the limit $k_B\rightarrow 0$. This corresponds to a change in the algebraic structure associated to the variable space. 

Secondly, the scaling $N_A\rightarrow \infty$ and its consequences on the
arithmetic process of counting affect the degenerations of energy levels, so it relates to the phenomenon of exponential degenerations. For instance, one can consider the Boltzmann
formula for the microcanonical ensemble $S_{B}=k_{B}\cdot\ln\Omega$. The number of microstates
$\Omega$ is the cardinality of the set of microstates compatible
with some constraints, i.e. fixed energy and particle number. Non-trivial
situations arise if one assumes that $k_{B}\rightarrow0$ implies
a different way of counting. The cardinality $\Omega$ thereby depends
on $k_{B}$ and exponential degenerations 
\begin{equation}
\Omega(k_{B})\sim\exp\left(\frac{S}{k_{B}}\right)\label{eq: exponential degenerations}
\end{equation}
make the tropical method non-trivial. 

A possible connection can also be drawn to the definition of temperatures
in small systems. At least two proposals have received attention,
that are Boltzmann and Gibbs (or Hertz) temperatures \cite{Frenkel2015},
and there is still much debate on which definition should be adopted
\cite{DunkelHilbert2014,Frenkel2015}. In both cases,
they are deduced from entropy via ${\displaystyle T=\left(\frac{\partial S}{\partial E}\right)^{-1}}$.
One has the Gibbs temperature $T_{G}$ if $S=S_{G}(E):=k_{B}\cdot\ln\Omega(E)$
is Gibbs' entropy, where $\Omega(E)$ is the number (or in general
a measure) of microstates with energy $\tilde{E}\leq E$. The Boltzmann
temperature comes from Boltzmann's entropy ${\displaystyle S_{B}(E)=k_{B}\cdot\ln\omega(E)}$,
where ${\displaystyle \omega(E)=\frac{\partial\Omega(E)}{\partial E}}$
is the number of states with energy $\tilde{E}=E$. $T_{G}$ is positive
while $T_{B}$ can assume negative values, for example in cases of
bounded spectrum. In the limit $k_{B}\rightarrow0$, assuming the
scaling (\ref{eq: exponential degenerations}) and additional hypotheses,
e.g. non-vanishing heat capacity, the two definitions
can coincide. 
When $T_{B}$ equals $T_{G}$ at $k_{B}=0$, only boundary terms in
$\Omega(E)$ are relevant for temperature, since they represent microstates
defining $\omega(E)$. Boltzmann's entropy is assumed real, then $\omega(E)\geq0$
for all $E$ and $\Omega(E)$ increases with $E$. So these boundary
terms are the dominant ones. 

The correspondence between expressions ``$\leq E$'' and ``$=E$''
in previous statements is a focal point in our discussion. Indeed,
we will consider \emph{microsystems} defined by fixed energy and associated
statistical data (e.g., degeneration of the energy level). On the
other hand, we will denote collections of more microsystems as \emph{macrosystems}.
Hence, the previous observation relates microsystems to certain macrosystems
and will be used to realize a tropical correspondence. 

Before we proceed to this issue, it is worth making a few remarks on the meaning of the tropical limit at the fundamental (set-theoretic) level. The change in the counting process also prompts the use of a different notion of cardinality or different notion of sets.
In this regard, it is worth noting that the limit $k_{B}\rightarrow0$
can be interpreted as a fuzzyfication, see \cite{Zadeh1965,Klir1988}
for a detailed introduction about fuzzy mathematics and \cite{Kahraman2016}
for its latest developments. Fuzzy variables for statistical purposes
has been used in theoretical computer science \cite{Yasuda2012}.
For example, the fuzzy $c$-means in the context of clustering (see
e.g. \cite{Miyamoto1997}) have a direct analogue in our statistical
setting, as can be noted identifying $m-1$ with $k_{B}$ in \cite{Miyamoto1997}. 

Boltzmann constant has also a key role in representing information.
For example, if $x$ in (\ref{eq: tropical limit}) is interpreted
as the number of available digits in a certain representation of a
number, then the scaling ${\displaystyle x\rightarrow\frac{x}{k_{B}}}$
describes a generalized change of base. In this aspect, Boltzmann
constant is recognized as a bridge between Gibbs and Shannon entropy
and, more generally, between thermodynamics and information theory
\cite{Landauer1961}. 

In all these perspectives, what is left in the limit $k_{B}\rightarrow0$
is the underlying relational order. We explore some aspects of this relation in the following section.

\section{\label{sec: Filters and ultrametricity} Filters and ultrametricity}

One of the most important geometric aspects of the tropical algebra
is ultrametricity. An ultrametric on a set $\Omega$ is a metric where
the triangle inequality is tropicalized, i.e. 
\begin{equation}
d(x,y)\leq d(x,z)\oplus d(z,y)=\max\left\{ d(x,z),d(z,y)\right\} ,\quad x,y,z\in\Omega.\label{eq: ultrametric, a}
\end{equation}
Ultrametricity is a phenomenon that characterizes many hierarchical
models. Accordingly, it has important applications in statistical
and complex systems. For example, ultrametricity was recognized in
spin glasses \cite{RTV1986} and some models have been proposed for
a better understanding of this pattern, for example by means of the $p$-adic metric
\cite{Koblitz1977,Khrennikov1994,Grossman1989}. We remind that, if $\Omega\equiv\mathbb{Q}$ and $p$ is a prime number, then the $p$-adic norm 
\begin{equation}
||x||_{p}=p^{-a}\Leftrightarrow x=p^{a}\cdot\frac{q}{r},\,\mathrm{gcd}(q,p)=\mathrm{gcd}(r,p)=1,\label{eq: p-adic norm}
\end{equation}
with the assumption $||0||_{p}=0$, induces the $p$-adic ultrametric
$||x-y||_{p}$ on $\mathbb{Q}$. 

For our purposes, the poset $(\mathcal{P}(\Omega),\subseteq)$ of the 
subsets of a set $\Omega$ ordered by inclusion has a cardinal
role. In such a case, the duality between filters and ideals is also
restated as follows. 
\begin{lemma}
\label{lem: ideals/filters power set case} Given a set $\Omega$,
$\mathcal{I}\subseteq\mathcal{P}(\Omega)$ is an ideal if and only
if $\mathcal{F}:=\{\Omega\backslash A:\,A\in\mathcal{I}\}$ is a filter. 
\end{lemma}
\begin{proof} The proof is a straightforward consequence of the definitions of filter and ideal, see also  \cite{HrbacekJech1999}. 
\end{proof}
There is a connection between metric and topology, which is expressed
by the fact that the set $\mathcal{B}=\{S(x_{0},r):,\,x_{0}\in\Omega,\,r\in\mathbb{R}_{+}\}$ of open balls $S(x_{0},r):=\{x\in\Omega:\,d(x,x_{0})<r\}$ for the
metric $d$ is a base for a topology on $\Omega$. In the particular
case of ultrametrics, a similar relation can be found with filters
(or ideals). For notational convenience, let us denote  $\mathbb{R}_{+}^{\wedge}:=\{x\in\mathbb{R}:\,x>0\}\cup\{+\infty\}$ and $d(x_{0},\cdot):=\left\{ d(x_{0},x):\,x\in\Omega\right\}$. 
\begin{proposition}
\label{prop: filters and ultrametrics} Let $d$ be an
ultrametric on a set $\Omega$. Then, the set $\mathcal{B}:=\{S(x_{0},r):\,x_{0}\in\Omega,\,r\in d(x_{0},\cdot)\}$ of ultrametric balls $S(x_{0},r):=\{x\in\Omega:\,d(x,x_{0})\leq r\}$ is an ideal base. Moreover, the resulting ideal is proper if and only if $d(\Omega^{2})$ is grounded as a sublattice of $\mathbb{R}_{>0}$,
i.e. $\max\left\{ d(x,y):\,x,y\in\Omega\right\} $ does not exists.
Vice versa, if $\mathcal{I}$ is an ideal on ${\displaystyle \Omega=\bigcup_{A\in\mathcal{I}}A}$,
then the function 
\begin{equation}
(x,y)\mapsto d(x,y):=\left\{ \begin{array}{cc}
0, & x=y\\
\inf\left\{ \mathfrak{d}(A):\,A\in\mathcal{I},\,\{x,y\}\subseteq A\right\} , & x\neq y
\end{array}\right.\label{eq: ideal to ultrametric}
\end{equation}
is an ultrametric on $\Omega$, for each decreasing function $\mathfrak{d}:\,(\mathcal{I},\subseteq)\longrightarrow(\mathbb{R}_{+}^{\wedge},\leq)$
such that $\inf\mathfrak{d}(\mathcal{I})>0$. 
\end{proposition}
\begin{proof} For the sake of clarity, the proof is presented in Appendix \ref{sec: Appendix A}.  \end{proof}

It should be noted that previous proposition does not hold for a general metric. For example, let 
\begin{equation}
\Omega:={\displaystyle \left\{ \frac{1}{n}:\,n\in\mathbb{N}\right\} }\cup \left\{ 3-\frac{1}{m}:\,m\in\mathbb{N}\right\}
\label{eq: no-example general metric}
\end{equation}
with the usual euclidean metric $d_{E}$. Note that $d_{E}(\Omega^{2})$ is grounded since $\sup d_{E}(\Omega^{2})=3\notin d_{E}(\Omega^{2})$. So $1=d_{E}(1,2)$ belongs to both $d_{E}(1,\cdot)$ and $d_{E}(2,\cdot)$. Hence, one has $S\left(1,1\right)\cup S\left(2,1\right)=\Omega$ that is not contained in any ball of the type $S(x_{0},r)$ with $x_{0}\in\Lambda$ and $r\in d_{E}(x_{0},\cdot)$. 

We remark that the function (\ref{eq: ideal to ultrametric}),
$\mathfrak{d}$ is allowed to assume the value $+\infty$. In such
a case, if $\mathfrak{d}(A)=+\infty$ for all $A\in\mathcal{I}$ such
that $\{x,y\}\subseteq A$, then $x$ and $y$ are at infinite distance.
This can be an interesting eventuality, but if one wants to avoid
it one can use the function $\mathfrak{g}:=g\circ\mathfrak{d}$ instead
of $\mathfrak{d}$, where 
\begin{equation}
g(x):=\left(1+(x)^{-1}\right)^{-1}
\label{eq: deinfinitation}
\end{equation} for
all $x\in\mathbb{R}_{+}^{\wedge}$. So $\mathfrak{g}$ is decreasing,
bounded by $1$ from above and strictly positive since $0<\inf\mathfrak{d}(\mathcal{I})\leq\mathfrak{d}(A)$
for all $A\in\mathcal{I}$.  

If one starts with an ultrametric to get an associated ideal, the
corresponding filter (Lemma \ref{lem: ideals/filters power set case})
can be used to recover the original ultrametric. Indeed, one has the
following 
\begin{corollary}
\label{cor: filter to ultrametric and back} If $\mathcal{F}$ is
a filter on $\Omega$ and ${\displaystyle \emptyset=\bigcap_{G\in\mathcal{F}}G}$,
then the function 
\begin{equation}
(x,y)\mapsto D(x,y):=\left\{ \begin{array}{cc}
0, & x=y\\
\inf\left\{ \mathfrak{d}(G):\,G\in\mathcal{F},\,\{x,y\}\cap G=\emptyset\right\} , & x\neq y
\end{array}\right.\label{eq: filter to ultrametric}
\end{equation}
is an ultrametric on $\Omega$, for each increasing function $\mathfrak{d}:\,(\mathcal{F},\subseteq)\longrightarrow(\mathbb{R}_{+}^{\wedge},\leq)$
such that $\inf\mathfrak{d}(\mathcal{F})>0$. Moreover, if $d$ is
an ultrametric on $\Omega$, $\mathcal{I}_{d}$ is the ideal generated
by ultrametric balls, $\mathcal{F}_{d}$ is the corresponding filter,
then $D_{d}=d$ if 
\begin{equation}
\mathfrak{d}(F)=\mathfrak{d}_{d}(F):=\sup\left\{ d(x,y):\,x,y\notin F\right\} .\label{eq: ultradiameter}
\end{equation}
\end{corollary}
\begin{proof} 
Let $\mathcal{F}$ and $\mathfrak{d}$ be as in the claim and say
$\mathfrak{D}(x,y):=\left\{ \mathfrak{d}(G):\,G\in\mathcal{F},\,\{x,y\}\cap G=\emptyset\right\} $,
$x\neq y$. Then, consider the corresponding ideal $\mathcal{I}$
as in Lemma \ref{lem: ideals/filters power set case} and the function
$\mathfrak{d}_{\mathcal{I}}:\,\mathcal{I}\longrightarrow\mathbb{R}_{+}^{\wedge}$
defined as $\mathfrak{d}_{\mathcal{I}}(A):=\mathfrak{d}(\Omega\backslash A)$.
In particular, $\mathfrak{d}_{\mathcal{I}}$ is decreasing since $\mathfrak{d}$
is increasing, and $\mathfrak{d}(\mathcal{F})=\mathfrak{d}_{\mathcal{I}}(\mathcal{I})$
implies $\inf\mathfrak{d}_{\mathcal{I}}(\mathcal{I})=\inf\mathfrak{d}(\mathcal{F})>0$.
From these data one can produce the mapping $\mathfrak{D}_{\mathcal{I}}$
and the function $d$ as in Proposition \ref{prop: filters and ultrametrics}.
By definitions, one has a correspondence between $G\in\mathcal{F},\,\{x,y\}\cap G=\emptyset$
and $\Omega\backslash G\in\mathcal{I},\,\{x,y\}\subseteq\Omega\backslash G$,
hence $\mathfrak{D}(x,y)=\mathfrak{D}_{\mathcal{I}}(x,y)$ and $D(x,y)=\inf\mathfrak{D}(x,y)=d(x,y)$.
So $D$ is an ultrametric. 

Now let $d$ be an ultrametric on $\Omega$, $\mathcal{I}_{d}$ the
ideal generated by the base of ultrametric balls and $\mathcal{F}_{d}$
the corresponding filter. Let us consider the function $\mathfrak{d}_{d}$
as in (\ref{eq: ultradiameter}). The corresponding ultrametric $D(x,y)$
obtained from (\ref{eq: filter to ultrametric}) coincides with $d$.
In fact, if $x\neq y$ and $\{x,y\}\cap G=\emptyset$, then one has
$d(x,y)\leq\mathfrak{d}_{d}(G)$ thus $d(x,y)\leq D(x,y)$. Moreover,
$\{x,y\}\cap\left(\Omega\backslash S(x,d(x,y))\right)=\emptyset$,
so $D(x,y)\leq\sup\left\{ d(u,v):\,u,v\notin\left(\Omega\backslash S(x,d(x,y))\right)\right\} =\sup\left\{ d(u,v):\,u,v\in S(x,d(x,y))\right\} =d(x,y)$.
Hence, $D(x,y)=d(x,y)$.  
\end{proof}

\subsection{\label{subsec: Some remarks on the conditions in Proposition [filters and ultrametrics} Discussion on the conditions in Proposition \ref{prop: filters and ultrametrics}}

An example of the relation between filters, ideals and ultrametric
is the one of $p$-adic norm (\ref{eq: p-adic norm}). One has $d(\mathbb{Q}^{2})\subseteq\{0\}\cup\{p^{n}:\,n\in\mathbb{Z}\}$
by definition. Vice versa, for all $n\in\mathbb{Z}$, one has $d(p^{-n},0)=||p^{-n}||_{p}=p^{n}$
and $d(0,0)=0$. Thus $d(\mathbb{Q}^{2})=\{0\}\cup\{p^{n}:\,n\in\mathbb{Z}\}$
is grounded and the set of $p$-adic balls generates a proper ideal.
On the other hand, the filter generated by $p$-adic balls and the
function $\mathfrak{d}_{p}(G):=\sup\left\{ ||x-y||_{p}:\,x,y\notin G\right\} $
returns the $p$-adic ultrametric $D_{p}(x,y):=||x-y||_{p}$. 

The condition $\inf\mathfrak{d}(\mathcal{I})>0$ (respectively, $\inf\mathfrak{d}(\mathcal{F})>0$) in Proposition \ref{prop: filters and ultrametrics} (respectively, Corollary \ref{cor: filter to ultrametric and back}) is due to the choice $\max$ for the tropical sum in (\ref{eq: ultrametric, a}) and can be relaxed with a different presentation of $\oplus$, as it will be shown in next section. 

The condition is sufficient to avoid the degeneration of $d$. If this request is not satisfied, degeneration
could occur. For example, let us take $\Omega=\mathbb{N}$ and $\mathcal{I}_{\mathrm{fin}}=\{A\subseteq\mathbb{N}:\,\#A<\infty\}$,
that is the ideal of finite subsets of $\mathbb{N}$. Let ${\displaystyle \mathfrak{d}(A):=1-\sum_{\alpha\in A}\frac{1}{2^{\alpha}}}$,
$A\in\mathcal{I}$. Note that $\mathfrak{d}(\mathcal{I})\subseteq\mathbb{R}_{>0}$
and $\mathfrak{d}$ is decreasing, but sets $A_{n}:=[n]=\{1,\dots,n\}$,
$n\geq3$, satisfy $\{1,2\}\in A_{n}$, $A_{n}\in\mathcal{I}_{\mathrm{fin}}$
and ${\displaystyle \mathfrak{d}(A_{n})=1-\sum_{i=1}^{n}\frac{1}{2^{i}}}\underset{n\rightarrow\infty}{\longrightarrow}0$.
So $d(1,2)=0$. Anyway, this condition is not necessary. For example,
in the $p$-adic case (\ref{eq: p-adic norm}) one has $\mathfrak{d}(\mathbb{Q}\backslash S(x,p^{-n}))=\sup\left\{ ||x-y||_{p}:\,||x-y||_{p}\leq p^{-n}\right\} =p^{-n}\underset{n\rightarrow+\infty}{\longrightarrow}0$.
Thus $\inf\mathfrak{d}(\mathcal{F})=0$, but the resulting $D_{p}(x,y)=||x-y||_{p}$
is not degenerate. 

The monotony condition on $\mathfrak{d}$ is sufficient but not necessary
too. For example, $\mathfrak{d}(F)=\sup\{d(x,y):\,x,y\notin F\}$
in Corollary \ref{cor: filter to ultrametric and back} is not increasing,
in fact it is decreasing. Nevertheless, it generates an ultrametric
since it is derived from an ultrametric. In general, non-monotone
functions do not return ultrametrics. An example is the filter $\mathcal{F}_{\mathrm{fin}}$
associated to $\mathcal{I}_{\mathrm{fin}}$ with any function $\mathfrak{d}:\,\mathcal{F}_{\mathrm{fin}}\longrightarrow\mathbb{R}$
such that $\mathfrak{d}(F)>1+\mathfrak{d}(G)$ for all $F,G\in\mathcal{F}$
with $\{1,2\}\cap F=\emptyset$ and $\{1,2\}\cap G\neq\emptyset$.
Such functions are not increasing. One also finds $d(1,3)\leq\mathfrak{d}(\mathbb{N}\backslash\{1,3\})$
and $d(2,3)\leq\mathfrak{d}(\mathbb{N}\backslash\{2,3\})$, thus $\max\{d(1,3),d(2,3)\}\leq\max\{\mathfrak{d}(\mathbb{N}\backslash\{1,3\}),\mathfrak{d}(\mathbb{N}\backslash\{2,3\})\}\leq d(1,2)-1$.
Hence $\max\{d(1,3),d(2,3)\}<d(1,2)$ and ultrametric triangle inequality
does not hold. 

\subsection{\label{subsec: The case of finite sets} The case of finite sets}

The image $d(\Omega^{2})$ is not grounded in case of finite sets
$\Omega$, so the resulting condition returns the trivial filter (and
ideal) $\mathcal{P}(\Omega)$. Anyway, filters on a finite set $\Omega$
are easily described by the following well-known result (see e.g. \cite{Mendelson2015}). 
\begin{proposition}
\label{prop: (ultra-)filters in finite case} If $\#\Omega<\infty$
and $\mathcal{F}$ is a filter on $\Omega$, then there exists a unique
non-empty subset $\zeta\subseteq\Omega$ such that ${\displaystyle \mathcal{F}=\left\{ \Theta\subseteq\Omega:\,\zeta\subseteq\Omega\right\} }$.
Moreover, $\mathcal{F}$ is an ultrafilter if and only if $\#\zeta=1$. 
\end{proposition}

Previous proposition means that filters on finite sets are principal. Non-principal
(ultra-)filters appear for infinite sets and play an important role in
mathematical logic through the concept of $\mathcal{F}$-limit \cite{Mendelson2015}.
On the other hand, principal filters satisfies ${\displaystyle \zeta(\mathcal{F})=\bigcap_{H\in\mathcal{F}}H\neq\emptyset}$.
In that case, if $z\in\zeta(\mathcal{F})$, then $\mathfrak{D}(z,y)=\emptyset$
in Corollary \ref{cor: filter to ultrametric and back}, so it is
consistent to assign $d(z,y)=\inf\emptyset:=\sup\mathbb{R}_{+}^{\wedge}=+\infty$.
If one acts via $g$ in (\ref{eq: deinfinitation}),
the resulting ultrametric is $g(d(x,y))\in[0,1]$, that vanishes if
and only if $x=y$ and equals $1$ if and only if $\{x,y\}\cap\zeta(\mathcal{F})\neq\emptyset$. 

The importance of filters on finite sets also lies in its
connections with the statistical amoeba formalism developed in \cite{AK2016b}.
The instability domain for a statistical amoeba is the locus of points
$\boldsymbol{x}$ in the parameter space $\mathbb{R}^n$ where the family 
\begin{equation}
\mathcal{N}_{k}(\boldsymbol{x}):=\left\{ \mathcal{I}\subseteq[N]:\,\#\mathcal{I}=k,\,\mathcal{Z}_{k}(\mathcal{I};\boldsymbol{x})<0\right\} \label{eq: family negative weighted set statistical amoeba}
\end{equation}
 has maximal cardinality, where ${\displaystyle k<\frac{N+1}{2}}$
and 
\begin{equation}
{\displaystyle \mathcal{Z}_{k}(\mathcal{I};\boldsymbol{x}):=-\sum_{\alpha\in\mathcal{I}}e^{f_{\alpha}(\boldsymbol{x})}+\sum_{\beta\notin\mathcal{I}}e^{f_{\beta}(\boldsymbol{x})}}.\label{eq: weight set statistical amoeba}
\end{equation}
These quantities were introduced in order to study real points where the partition function becomes singular and explore the associated metastability. 
In particular, the results in \cite{AK2016b} mean that $\mathcal{N}_{k}(\boldsymbol{x})$
is induced by an ultrafilter if and only if $\boldsymbol{x}$ belongs to the
instability domain $\mathcal{D}_{k-}$.

\section{\label{sec: Monoid homomorphisms and a set/element correspondence}
Monoid homomorphisms and a set/element correspondence}

The combinatorial data that remain in the tropical limit
of a statistical system are based on an order relation, which is captured by the concepts of filters and
ideals. In the tropical perspective, these order relations separate the notions of a 
function and of its expressions. Indeed, let us consider tropical
polynomials, that are algebraic expressions of the form 
\begin{equation}
\bigoplus_{I\in\mathbb{I}}a_{I}\odot\boldsymbol{X}^{\odot I}\label{eq: tropical polynomial}
\end{equation}
where $\mathbb{I}$ is a finite set of multi-indices $I:=(i_{1},\dots,i_{n})$,
$\{a_{I}\}_{I\in\mathbb{I}}$ are coefficients and $\boldsymbol{X}^{\odot I}:=X_{1}^{\odot i_{1}}\odot\dots\odot X_{n}^{\odot i_{N}}$.
The expression (\ref{eq: tropical polynomial}) is equivalent to 
\begin{equation}
{\displaystyle \left(\bigoplus_{I\in\mathbb{I}}a_{I}\odot\boldsymbol{X}^{\odot I}\right)\oplus\underset{{\scriptstyle j\mbox{ times}}}{\underbrace{a_{H}\odot\boldsymbol{X}^{\odot H}}}}\label{eq: tropical polynomial, 2}
\end{equation}
for any $j\in\mathbb{N\cup}\{0\}$ and $a_{H}\in\{a_{I}\}_{I\in\mathbb{I}}$.
Such a redundancy in the description of algebraic functions in many
different but equivalent forms is distinctive of tropical algebra
and can be referred as 	\emph{tropical symmetry}. 

So, if $F(\boldsymbol{x}):\,\Omega\longrightarrow\Lambda$ is a tropical
expression from a set $\Omega$ in a tropical monoid $\Lambda$, then
all the expressions in the set $\mathcal{I}_{F(\boldsymbol{x})}:=\left\{ G(\boldsymbol{x}):\,G(\boldsymbol{x})\preceq F(\boldsymbol{x})\right\} $
satisfy $F(\boldsymbol{x})\oplus\mathcal{I}_{F(\boldsymbol{x})}=\{F(\boldsymbol{x})\}$.
The set $\mathcal{I}_{F(\boldsymbol{x})}$ is a principal ideal for
the poset $\Lambda$. Then, the equivalence of (\ref{eq: tropical polynomial})
and (\ref{eq: tropical polynomial, 2}) can be stated using the language
of filters and ideals. 

For tropical monoids, an erasing element $\top$ is the maximum for
the associated order, so grounded monoids are grounded posets. As already remarked in the Introduction, monoids give a special role to the neutral element $\infty$ of $\oplus$, which is an extremal for the corresponding order. The
fact that $\infty$ is the only extremal element in grounded
posets is fundamental in the symmetry breaking between $\max$ and
$\min$. Lemma \ref{lem: ideals/filters power set case} basically
identifies ideals and filters on the power set $\mathcal{P}(\Omega)$,
which is not a grounded monoid since it has both a minimum $\emptyset$
and a maximum $\Omega$. In case of grounded monoids, the relation
between filters and ideal is more involved and has relevant physical
interpretations connected with the $\max$/$\min$ duality. 

A first implication is that the absence of an erasing element is an
obstruction for the definition of a dual operation $\tilde{\oplus}$
on $\Lambda$ such that $(\Lambda,\tilde{\oplus})$ is a tropical
monoid and the order $\succeq$ induced by $\tilde{\oplus}$ is the
opposite of $\preceq$. For instance, that gives a split between the
$\max$-plus and the $\min$-plus algebras. In applications, it distinguishes
different constraints on physical systems, since the stability of a configuration
is often determined by extremality conditions. This is the case of
a system defined by a Lagrangian $\mathcal{L}:=\frac{1}{2}m||\dot{\boldsymbol{x}}||^{2}-\phi(\boldsymbol{x})$,
local minima for the potential $\phi(\boldsymbol{x})$ are stable
equilibrium points, while local maxima are unstable equilibrium points.
This symmetry breaking leads to an orientation expressed by the order
relation $\preceq$, which is the ``tropical skeleton'' discussed
in the Introduction. 

The $\max$\textbackslash{}$\min$ splitting also concerns the ultrametric
triangle inequality (\ref{eq: ultrametric, a}). Its form with $\oplus=\min$
is 
\begin{equation}
\tilde{d}(x,z)\oplus\tilde{d}(z,y)=\min\left\{ \tilde{d}(x,z),\tilde{d}(z,y)\right\} \leq\tilde{d}(x,y),\quad x,y,z\in\Omega.\label{eq: ultrametric, b}
\end{equation}
Both (\ref{eq: ultrametric, a}) and (\ref{eq: ultrametric, b}) are
summarized in the expression 
\begin{equation}
\tilde{d}(x,z)\oplus\tilde{d}(z,y)\oplus\tilde{d}(x,y)=\tilde{d}(x,z)\oplus\tilde{d}(x,y),\quad x,y,z\in\Omega.\label{eq: ultrametric, algebraic reduction}
\end{equation}
On the other hand, filters and ideals distinguish these forms. The next result is a straightforward consequence of previous observations and stresses the effect of dual presentations (orders) of grounded posets on the geometry of the system (non-degeneracy conditions for the ultrametric). 
\begin{corollary}
\label{cor: conjugate form ultrametric} Let $\mathcal{I}$ be an
ideal on $\Omega$, ${\displaystyle \bigcup_{A\in\mathcal{I}}A=\Omega}$
and $\mathfrak{d}:\,\mathcal{I}\longrightarrow\mathbb{R}_{+}^{\wedge}$
be any increasing function. Then the function 
\begin{equation}
(x,y)\mapsto d(x,y):=\left\{ \begin{array}{cc}
0, & x=y\\
\sup\left\{ \mathfrak{d}(A):\,A\in\mathcal{I},\,\{x,y\}\subseteq A\right\} , & x\neq y
\end{array}\right.\label{eq: ideal to ultrametric, b}
\end{equation}
is an ultrametric in the form (\ref{eq: ultrametric, b}). In particular,
if $d$, $\mathcal{I}_{d}$, $\mathcal{F}_{d}$ and $\mathfrak{d}_{d}$
are as in Corollary \ref{cor: filter to ultrametric and back}, then
$\mathfrak{d}_{d}$ induces the presentations (\ref{eq: ultrametric, a})
and (\ref{eq: ultrametric, b}) for the same ultrametric structure,
when applied to $\mathcal{I}_{d}$ and $\mathcal{F}_{d}$ respectively. 
\end{corollary}
\begin{proof}
Let $\mathcal{I}$ be an ideal and consider any increasing positive
function $\mathfrak{d}:\,(\mathcal{I},\subseteq)\longrightarrow(\mathbb{R}_{+}^{\wedge},\leq)$.
One can see that $d$ is symmetric and $\mathfrak{D}(x_{1},x_{2})\neq\emptyset$,
so $d(x_{1},x_{2})$ is strictly positive if $x_{1}\neq x_{2}$. One also gets $\min\{d(x,y),d(y,z)\}\leq d(x,z)$
by using the same arguments as in Proposition \ref{prop: filters and ultrametrics}.
From Corollary \ref{cor: filter to ultrametric and back}, $\mathfrak{d}_{d}$
returns the form (\ref{eq: ultrametric, a}) for the ultrametric structure
of $d$ when (\ref{eq: ideal to ultrametric}) is applied to the ideal $\mathcal{I}_{d}$.  
\end{proof}

\subsection{\label{subsec: Duality and the role of linear orders} Duality and the role of linear orders}

We now focus on a class of tropical algebras that offers a simple but
flexible framework for physical applications. So, we consider a
tropical monoid $\Lambda$ and introduce 
\begin{equation}
\begin{array}{cc}
\mathcal{\iota}(y):=\left\{ x\in\Lambda:\,x\prec y\right\} , & \mathcal{I}_{\Lambda}:=\left\{ \mathcal{\iota}(y):\,y\in\Lambda\right\} \\
\phi(y):=\left\{ x\in\Lambda:\,y\preceq x\right\} , & \mathcal{F}_{\Lambda}:=\left\{ \phi(y):\,y\in\Lambda\right\} 
\end{array}.\label{eq: principal filters and dual}
\end{equation}
In principle the order $\preceq$ can be general. In fact, this freedom
allows one to deal with a broader class of orders, including set-theoretic
inclusion $\subseteq$. Anyway, the use of filters of the form (\ref{eq: principal filters and dual}) gives a special role to totally (or linearly) ordered sets. Indeed, one can characterize totally ordered sets through the following simple property. 
\begin{proposition}
\label{prop: homomorphism and linear order} The mapping $\iota:\,(\Lambda,\oplus,\infty)\longrightarrow(\mathcal{I}_{\Lambda},\cup,\emptyset)$
defined in (\ref{eq: principal filters and dual}) is a monoid homomorphisms
if and only if $(\Lambda,\oplus,\infty)$ is totally ordered. The
mapping $\phi:\,(\Lambda,\oplus,\infty)\longrightarrow(\mathcal{F}_{\Lambda},\cap,\Lambda)$
is always a monoid homomorphism.
\end{proposition}
\begin{proof} $\emptyset=\iota(\infty)$ belongs to $\mathcal{I}_{\Lambda}$ and
it is a neutral element for $\cup$. So $(\mathcal{I}_{\Lambda},\cup,\emptyset)$
is a tropical monoid if and only if, for all $x,y$ in $\Lambda$,
there exists $z\in\Lambda$ such that $\iota(x)\cup\iota(y)=\iota(z)$.
Let us assume that $\Lambda$ is not totally ordered, then there exist
$x,y\in\Lambda$ such that $x\oplus y\notin\{x,y\}$. In particular
$x\neq y$, $x\notin\iota(y)$ and $y\notin\iota(x)$. Let us take
$z\in\Lambda$ such that $\iota(x)\cup\iota(y)=\iota(z)$. From previous
observation one has $\{x,y\}\cap\iota(z)=\emptyset$. So $x\in\iota(x\oplus y)$
but $x\notin\iota(z)=\iota(x)\cup\iota(y)$, hence $\iota(x\oplus y)\neq\iota(x)\cup\iota(y)$.
Thus $\iota$ is not a monoid homomorphism. If instead $\Lambda$
is a linearly ordered set, then it is easy to see that $\iota(a)\cup\iota(b)=\iota(\sup\{a,b\})=\iota(a\oplus b)$,
so $(\mathcal{I}_{\Lambda},\cup,\emptyset)$ is a tropical monoid
and $\iota$ is a monoid homomorphism. 

Now let us consider the tropical monoid $(\mathcal{F}_{\Lambda},\cap,\Lambda)$.
If $z\in\phi(x)\cap\phi(y)$ then $x\preceq z$ and $y\preceq z$,
hence $x\oplus y\preceq z$ and $z\in\phi(x\oplus y)$. Thus $\phi(x)\cap\phi(z)\subseteq\phi(x\oplus y)$.
If $z\in\phi(x\oplus y)$ then $x\preceq x\oplus y\preceq z$ so $z\in\phi(x)$.
Similarly $z\in\phi(y)$, hence $z\in\phi(x)\cap\phi(y)$. This means
that $\phi(x)\cap\phi(y)=\phi(x\oplus y)$. Moreover $\phi(\infty)=\Lambda$,
hence $\phi:\,(\Lambda,\oplus,\infty)\longrightarrow(\mathcal{F}_{\Lambda},\cap,\Lambda)$
is always a monoid homomorphism.  
\end{proof}

Set complementation $c:\,(\mathcal{P}(\Lambda),\cup,\emptyset)\longrightarrow(\mathcal{P}(\Lambda),\cap,\Lambda)$
is a monoid isomorphism. In the case of totally ordered lattices $\Lambda$,
the restriction $\tilde{c}:=c|_{\mathcal{I}(\Lambda)}$ acts as $\tilde{c}(\iota(x))=\phi(x)$
and is a monoid isomorphism that allows to identify $(\mathcal{I}_{\Lambda},\cup,\emptyset)$
and $(\mathcal{F}_{\Lambda},\cap,\Lambda)$ through $\phi=\tilde{c}\circ\iota$.
Thus, linearly ordered sets imply the compatibility of set complementation and algebraic characteristics. 

Proposition \ref{prop: homomorphism and linear order} also underlines a correspondence between elements and (a class of) subsets of a monoid $\Lambda$ that is compatible with tropical composition. As already mentioned in the Introduction, such a set/element correspondence is manifest in the physical identification of the tropical
limit with the scaling $N_{A}\rightarrow\infty$. In this context,
the set/element correspondence suggested in Section \ref{sec: Tropical limit and the role of Boltzmann constant} involves the tropical totally ordered
monoid $\mathbb{N}^{\vee}$ and the bijection between $N_{A}\in\mathbb{N}^{\vee}$
and $\phi(N_{A})=\{n\in\mathbb{N}^{\vee}:\,N_{A}\leq n\}\in\mathcal{F}_{\mathbb{N}}$ in (\ref{eq: principal filters and dual}). 

One can also look at the set of principal ideals 
\begin{equation}
\phi_{c}(y):=\left\{ x\in\Lambda:\,x\preceq y\right\} ,\quad\mathcal{J}_{\Lambda}:=\left\{\phi_{c}(y):\,y\in\Lambda\right\} .\label{eq: principal ideals}
\end{equation}
All posets $\phi_{c}$ have both a minimum $\bot$ and a maximum
$y$, hence none of them is grounded. Clearly, $\phi_{c}$ is never
a monoid homomorphism between $(\Lambda,\oplus,\bot)$ and $(\mathcal{P}(\Lambda),\cup,\emptyset)$
since $\phi_{c}(\bot)=\{\bot\}\neq\emptyset$. Anyway, it induces
a monoid homomorphism $\hat{\phi}_c:=(\Lambda,\oplus,\bot)\longrightarrow\left(\mathcal{P}(\Lambda\backslash\{\bot\}),\cup,\emptyset\right)$
by ``grounding'', i.e. considering $\hat{\phi}_c(y)=\phi_{c}(y)\backslash\{\bot\}$.
This gives the same result as $\phi$ in (\ref{eq: principal filters and dual})
with the opposite order $\succeq$. Proposition \ref{prop: homomorphism and linear order} shows that the slight difference between (grounded) $\prec$ in (\ref{eq: principal filters and dual}) and
(non-grounded) $\preceq$ in (\ref{eq: principal ideals}) brings
to different constraints in order to preserve algebraic properties.

\subsection{\label{subsec: Linear orders via homomorphisms} Linear orders via homomorphisms}

If the poset $\Lambda$ is not totally ordered, then
there could exist $y\in\Lambda$ such that $\phi_{c}(y)$ is not
grounded even if $\Lambda$ be a grounded monoid. In such a case,
there exist $y,\top,\top_{1}\in\Lambda$ such that $\max\phi_{c}(y)=\top\prec\top_{1}$.

In the meantime, non-totally ordered sets give one the freedom to
have several totally ordered substructures (chains) in the same framework.
A way to extract totally ordered sets from posets is given by monoid
homomorphisms from a totally ordered set $\Delta$. So, for a monotone
increasing $\vartheta:\,\Delta\longrightarrow\Lambda$ one can define
the maps 
\begin{equation}
\begin{array}{cc}
\iota_{\vartheta}(a):=\iota(\vartheta(a)), & \quad\phi_{\vartheta}(a):=\phi_c(\vartheta(a))\end{array}.\label{eq: ideals for filter presentation}
\end{equation}

Let us consider the totally ordered tropical monoid $\mathbb{R}_{\max}$
and an increasing map $\vartheta:\,\mathbb{R}^{\vee}\longrightarrow\Lambda$.
Note that one has $\iota_{\vartheta}(a)\cup(\iota_{\vartheta}(a)\cup\iota_{\vartheta}(b))=\iota_{\vartheta}(a)\cup\iota_{\vartheta}(b)$
for all $a,b\in\mathbb{R}^{\vee}$, hence $\iota_{\vartheta}(a)\subseteq\iota_{\vartheta}(a)\cup\iota_{\vartheta}(b)$.
Thus one has $\iota_{\vartheta}(a)\subseteq\iota_{\vartheta}(\max\{a,b\})$
and $\iota_{\vartheta}(b)\subseteq\iota_{\vartheta}(a\oplus b)$, which
implies $\iota_{\vartheta}(a)\cup\iota_{\vartheta}(b)\subseteq\iota_{\vartheta}(\max\{a,b\})$.
But $\iota_{\vartheta}(\max\{a,b\})\in\{\iota_{\vartheta}(a),\iota_{\vartheta}(b)\}$,
then $\iota_{\vartheta}(\max\{a,b\})\subseteq\iota_{\vartheta}(a)\cup\iota_{\vartheta}(b)$.
This means that $\iota_{\vartheta}(a\oplus b)=\iota_{\vartheta}(a)\cup\iota_{\vartheta}(b)$
holds for all increasing functions $\vartheta$ from $\mathbb{R}^{\vee}$
to $\Lambda$. The only condition needed to get a monoid homomorphism
is $\iota_{\vartheta}(-\infty)=\{-\infty\}$. Physically, one can say
that $\iota_{\vartheta}$ is a monoid homomorphism if and only if it
preserves the choice of the vacuum. 

Having remarked the opportunity to extract totally ordered sets from general posets using monoid homomorphisms, a natural question in the light of algebraic correspondence between elements and subsets is the extension of homomorphism from sets to their power sets. At this purpose, it is worth referring to a tropical monoid $(\Lambda,\oplus,\infty)$ as a \textit{almost complete lattice} if the extension 
\begin{equation}
\hat{\Lambda}:=\Lambda\cup\{\top\}\label{eq: almost complete, extension}
\end{equation}
with the relation $a\preceq\top$ for all $a\in\Lambda$ is a join-complete
semilattice, i.e. it admits arbitrary $\sup$ (and sums). 

For example, the
tropical monoid $(\Lambda,\oplus,\bot)=(\mathcal{P}_{\mathrm{fin}}(\mathbb{N}),\cup,\emptyset)$ of finite subsets of $\mathbb{N}$ is grounded since $\mathbb{N}$ is not finite. If $\{A_{n}\}$ is any collection of elements in $\mathcal{P}_{\mathrm{fin}}$ and ${\displaystyle \bigcup_{n}A_{n}}$ is not finite, then $\mathbb{N}=\sup A_{n}$ in $\mathcal{P}_{\mathrm{fin}}(\mathbb{N})\cup\{\mathbb{N}\}$.
So it is also almost complete with the extension $\top=\mathbb{N}$. 

On the other hand, let us take $(\mathbb{R}_{\star}^{\vee}:=\mathbb{R}^{\vee}\backslash\{0\},\sup,-\infty)$. One can consider
an extension $\mathbb{R}_{\star}^{\vee}\cup\{+\infty\}$ as in (\ref{eq: almost complete, extension}).
If there exists ${\displaystyle \omega:=\sup\{a\in\mathbb{R}_{\star}^{\vee}:\,a<0\}\in\mathbb{R}_{\star}^{\vee}\cup\{+\infty\}}$,
then $\omega$ cannot be neither negative, since it would be ${\displaystyle \omega<\frac{\omega}{2}<0}$,
nor positive, since ${\displaystyle 0<\frac{\omega}{2}<\omega}$ would
be an upper bound smaller than the least upper bound. So $\omega$
does not exists. 

Even if the distinction of complete semilattices and complete lattices is inessential from the perspective of order theory, it becomes relevant when one looks at the algebraic structure, including homomorphisms. These maps distinguish between an operation and the dual one. Also note that all join-complete semilattices are also almost complete. The following proposition clarifies some simple properties of the link between the tropical structure of $\Lambda$ and of its power set. 
 
\begin{proposition}
\label{prop: almost completeness} The tropical monoid $(\Lambda,\oplus,\infty)$
is an almost complete lattice if and only if $\iota(y)$ is an almost
complete lattice for all $y\in\Lambda$. Moreover, if $\psi:\,(\Lambda,\oplus,\infty)\longrightarrow(\tilde{\Lambda},\tilde{\oplus},\tilde{\infty})$ is a monoid homomorphism and $\phi(\Lambda)$ is grounded, then $\Lambda$ is grounded too. 
\end{proposition}
\begin{proof} Clearly $\iota(y)$ is an almost complete lattice if and only if $\phi_c(y)$
is a join-complete semilattice. Let $\Lambda$ be an almost complete lattice
and take its extension $\hat{\Lambda}$ as (\ref{eq: almost complete, extension}).
If $\mathcal{S}\subseteq\phi_c(y)$ then all elements $z\in\mathcal{S}$
satisfy $z\preceq y$, so $\sup\mathcal{S}$ exists in $\hat{\Lambda}$
and $\sup\mathcal{S}\preceq y$. In particular, $\sup\mathcal{S}\in\phi_c(y)$
and $\phi_c(y)$ is join-complete. Vice versa, let $\phi_c(x)$
be a join-complete semilattice for all $x\in\Lambda$ and take any $\mathcal{S}\subseteq\Lambda$.
If there exists $y\in\Lambda$ such that $y_{1}\preceq y$ for all
$y_{1}\in\mathcal{S}$, then $\mathcal{S}\subseteq\phi_c(y)$,
which is join-complete. Hence $\sup\mathcal{S}$ exists and it is
an element of $\phi_c(y)$. Otherwise, for all $y\in\Lambda$ there
exists $y_{1}\in\mathcal{S}$ such that $y_{1}\not\preceq y$. So
$\top\in\hat{\Lambda}$ is the only upper bound for $\mathcal{S}$,
hence it is $\sup\mathcal{S}=\top$. This is independent of the choice
of $\mathcal{S}$ among all subsets of $\Lambda$ such that $\sup\mathcal{S}\notin\Lambda$.
Finally, if $\top\in\mathcal{T}\subseteq\hat{\Lambda}$ then $\sup\mathcal{T}=\top$.
Thus $\hat{\Lambda}$ is a join-complete semilattice. 

Now, let us assume that $\psi:\,(\Lambda,\oplus,\infty)\longrightarrow(\tilde{\Lambda},\tilde{\oplus},\tilde{\infty})$ is a monoid homomorphism. If $\psi(\Lambda)$ is grounded, then for all $a\in\Lambda$ there exists $b\in\Lambda$ such that $\psi(a)\neq\psi(a)\tilde{\oplus}\psi(b)=\psi(a\oplus b)$, so $a\neq a\oplus b$. Hence $\Lambda$ is grounded.  
\end{proof}

Finally, next result extends the link in Proposition \ref{prop: almost completeness} to the case of monoid homomorphisms. 

\begin{proposition}
\label{prop: homomorphism totally ordered sets} Let $(\Delta,\max,-\infty)$
be a tropical monoid whose induced order is total. So, 
\begin{enumerate}
\item if $\psi:\,(\Delta,\max,-\infty)\longrightarrow(\Lambda,\oplus,\bot)$
is a monoid homomorphism, then $\psi(\Delta)$ is totally ordered. 
\item If $\phi_{\vartheta}\,(\Delta,\max,-\infty)\longrightarrow\left(\{A\subseteq\Lambda:\,\bot\in A\},\cup,\{\bot\}\right)$
in (\ref{eq: ideals for filter presentation}) is a monoid homomorphism,
then $\vartheta:\,\Delta\longrightarrow\Lambda$ is a monoid
homomorphism too. 
\item If $\Lambda$ is an almost complete lattice, then for every monoid
homomorphism $\psi:\,\Delta\longrightarrow\mathcal{P}(\Lambda)$
there exist a sublattice $\Delta_{0}\subseteq\Delta$ and monoid homomorphisms
$\vartheta:\,\Delta_{0}\longrightarrow\Lambda$ and $\bar{\iota}:\,\Delta_{0}\longrightarrow\mathcal{P}(\Lambda)$
such that $\bar{\iota}=\iota_{\vartheta}=\iota\circ\vartheta$ and $\phi_{\vartheta}$
covers $\psi$, namely $\psi(a)\subseteq\bar{\iota}(a)\cup\{\vartheta(a)\}=\phi_{\vartheta}(a)$ for all $\psi(a)$ having an upper bound. 
\end{enumerate}
\end{proposition}
\begin{proof}
For the sake of clarity, the proof is presented in Appendix \ref{sec: Appendix B}. 
\end{proof}

The last proposition identifies a simple connection between homomorphisms to a set and to  the associated power set. Statement 2. in Proposition \ref{prop: homomorphism totally ordered sets} describes the extension of a monoid homomorphism to subsets. Vice versa, statement 3. allows one to extract data from a monoid homomorphism to a power set and reduce to the underlying set. 

It should be remarked that, if $\Lambda$ is not an almost
complete lattice, statement 3. does not necessarily hold. Let us take $\Lambda\equiv\mathbb{R}_{\star}^{\vee}$, which is not almost complete, and a mapping $\psi:\,\mathbb{R}^{\vee}\longrightarrow\mathcal{P}(\mathbb{R}_{\star}^{\vee})$ defined by $\psi(a)=\{b\in\mathbb{R}_{\star}^{\vee}:\,b\leq a\}$.
If such a function $\vartheta$ would exist, then $\vartheta(0)\in\iota_{\vartheta}(\varepsilon)$
for all $\varepsilon>0$. In particular $\vartheta(0)\leq\sup\psi(\varepsilon)=\varepsilon$ for all $\varepsilon>0$. In the same way one can see that $\varepsilon=\sup\{\psi(\varepsilon)\}\leq\vartheta(0)$
for all $\varepsilon<0$. Thus the only possibility for $\vartheta(0)$
is $0$, which is not in $\mathbb{R}_{\star}^{\vee}$. The role of almost complete lattices when one looks at chains of a poset $\Lambda$, in light of Propositions \ref{prop: almost completeness} and \ref{prop: homomorphism totally ordered sets}, supports the choice of macrosystems (\ref{eq: principal filters and dual}) for the tropical correspondence.

\section{\label{sec: Perturbative tropical limit} Nested tropical limit}

Filters on finite sets, whose explicit form is remarked in Proposition \ref{prop: (ultra-)filters in finite case}, also appear in relation to a sort of ``perturbative" construction for the tropical limit in statistical physics. Indeed, let us take a partition function 
\begin{equation}
{\displaystyle \mathcal{Z}(\boldsymbol{x})=\sum_{\alpha=1}^{N}e^{f_{\alpha}(\boldsymbol{x})}}\label{eq: partition function}
\end{equation}
 relative to (free) energies $f_{\alpha}(\boldsymbol{x})$, $\alpha\in[N]$.
Let us introduce the \textit{nesting form of type}
$A$ of the partition function $\mathcal{Z}(\boldsymbol{x})$ is the
set of data generated by the base case 
\begin{eqnarray}
\mathfrak{M}_{0}(\boldsymbol{x}) & := & \left\{ \alpha\in[N]:\,f_{\alpha}(\boldsymbol{x})=\max\left\{ f_{\beta}(\boldsymbol{x}):\,\beta\in\beta\in[N]\right\} \right\} ,\nonumber \\
\mathfrak{S}_{0}(\boldsymbol{x}) & := & [N]\backslash\mathfrak{M}_{0},\nonumber \\
\mu_{0}(\boldsymbol{x}) & := & f_{\alpha_{0}}(\boldsymbol{x}),\quad\alpha_{0}\in\mathfrak{M}_{0},\nonumber \\
\nu_{0}(\boldsymbol{x}) & := & \#\mathfrak{M}_{0}(\boldsymbol{x})
\label{eq: nesting data, base}
\end{eqnarray}
and recursively extended as follows 
\begin{eqnarray}
\mathfrak{M}_{\ell}(\boldsymbol{x}) & := & \left\{ \alpha\in\mathfrak{S}_{\ell-1}(\boldsymbol{x}):\,f_{\alpha}(\boldsymbol{x})=\max\left\{ f_{\beta}:\,\beta\in\mathfrak{S}_{\ell-1}(\boldsymbol{x})\right\} \right\} ,\nonumber \\
\mathfrak{S}_{\ell}(\boldsymbol{x}) & := & [N]\backslash\left({\displaystyle \bigcup_{h=0}^{\ell}}\mathfrak{M}_{h}\right),\nonumber \\
\mu_{\ell}(\boldsymbol{x}) & := & f_{\alpha_{\ell}}(\boldsymbol{x}),\quad\alpha_{\ell}\in\mathfrak{M}_{\ell},\nonumber \\
\nu_{\ell}(\boldsymbol{x}) & := & \#\mathfrak{M}_{\ell}(\boldsymbol{x})
\label{eq: nesting data, iteration A}
\end{eqnarray}
up to $\ell=L$, which is the minimum integer such that $\mathfrak{S}_{L}(\boldsymbol{x})=\emptyset$.
The \textit{nesting form of type} $B$ of the partition function $\mathcal{Z}(\boldsymbol{x})$
is the nesting form of the partition function relative to energies
$-f_{\alpha}(\boldsymbol{x})$, $\alpha\in[N]$. 

The nesting form of type $B$ can be also obtained by substitution
of $\max$ with $\min$ in (\ref{eq: nesting data, base}), (\ref{eq: nesting data, iteration A}).
In the finite $N$ case that we are dealing with, it is defined as
\begin{eqnarray}
\mathfrak{m}_{\ell}(\boldsymbol{x}):=\mathfrak{M}_{L-\ell}(\boldsymbol{x}),\, & \mathfrak{s}_{\ell}(\boldsymbol{x}):=[N]\backslash\left({\displaystyle \bigcup_{h=0}^{\ell}}\mathfrak{m}_{h}\right),\nonumber \\
\kappa_{\ell}(\boldsymbol{x}):=\mu_{L-\ell}(\boldsymbol{x}),\, & \lambda_{\ell}(\boldsymbol{x}):=\nu_{L-\ell}(\boldsymbol{x})\label{eq: nesting data, iteration, finite, B} 
\end{eqnarray}
for $0\leq\ell\leq L$. The filter $\mathcal{F}:=\{B\subseteq[N]:\,\mathfrak{m}_{0}(\boldsymbol{x})\subseteq B\}$
is the set of sure events with respect to the usual probability ${\displaystyle {\displaystyle W}_{\alpha}=\frac{\#\left(\{\alpha\}\cap\mathfrak{m}_{0}(\boldsymbol{x})\right)}{\lambda_{0}}}$
also introduced in \cite{AK2015}. 

The data $\{\mu_{\ell}(\boldsymbol{x}),\nu_{\ell}(\boldsymbol{x}):\,\ell\in\{0,\dots,L\}\}$
can be visualized from the partition function in the following equivalent
form: 
\begin{eqnarray}
\mathcal{Z}(\boldsymbol{x}) & = & e^{\mu_{0}(\boldsymbol{x})}\cdot\left(\nu_{0}+e^{-\mu_{0}(\boldsymbol{x})}\cdot\sum_{\alpha\in\mathfrak{S}_{0}(\boldsymbol{x})}e^{f_{\alpha}(\boldsymbol{x})}\right)\nonumber \\
& = & e^{\mu_{0}(\boldsymbol{x})}\cdot\left(\nu_{0}+e^{\mu_{1}(\boldsymbol{x})-\mu_{0}(\boldsymbol{x})}\cdot\left(\nu_{1}+e^{-\mu_{1}(\boldsymbol{x})}\cdot\sum_{\alpha\in\mathfrak{S}_{1}(\boldsymbol{x})}e^{f_{\alpha}(\boldsymbol{x})}\right)\right)\nonumber \\
& = & e^{\mu_{0}(\boldsymbol{x})}\cdot\left(\nu_{0}+e^{\mu_{1}(\boldsymbol{x})-\mu_{0}(\boldsymbol{x})}\cdot\left(\nu_{1}+e^{\mu_{2}(\boldsymbol{x})-\mu_{1}(\boldsymbol{x})}\cdot\left(\nu_{2}+\dots\right.\right.\right. \nonumber \\
 & & \left.\left.\left.\left.\dots\cdot\left(\nu_{L-2}+e^{\mu_{L-1}(\boldsymbol{x})-\mu_{L-2}(\boldsymbol{x})}\cdot\right.\left(v_{L-1}+e^{\mu_{L}(\boldsymbol{x})-\mu_{L-1}(\boldsymbol{x})}\cdot\nu_{L}\right)\dots\right)\right)\right)\right).\label{eq: nested partition function} 
\end{eqnarray}

This method can be considered as an alternative perturbation expansion
of the tropical limit. Indeed, the function ${\displaystyle e^{-\frac{1}{k}}}$
is not analytic around $k=0$, thus standard perturbative tools (e.g.,
series expansion) fail in this case. 
\begin{proposition}
\label{prop: trivial naive perturbation at m>1} Zeroth and first
order corrections of $F(k_{B})$ correspond to data $\kappa_{0}$
and $\lambda_{0}$ respectively. The $m$-th order coefficient in
the Taylor expansion of $F$ near $k=0$ vanishes, for all $m\geq2$. 
\end{proposition}

\begin{proof} This is a simple calculation that is presented in Appendix \ref{sec: Appendix C}. 
\end{proof}

So zeroth and first order corrections to tropical free energy give
data $(\kappa_{0},\lambda_{0})$ relative to the first level in the
nesting form of $B$-type. They correspond to tropical free energy
and statistical prefactors $-\ln\lambda_{0}$ for entropy in \cite{AK2015}.
The standard perturbative approach allows one to get contributions
up to first order, i.e., relatively to first level of the nesting form. Nevertheless,
this process breaks out at higher levels. Contributions beyond the
first order are included in subsequent levels. One can apply series
expansion along with this nested structure to recover data (\ref{eq: nesting data, iteration, finite, B}). One
gets $\kappa_{\ell}$ (from the $0$-th order of the $(\ell+1)$-th
level) and $\lambda_{\ell}$ (from the first order of the $(\ell+1)$-th
level). 

If this construction is done with finitely many levels, it returns
a finite set of data. In such a case, the recursion as a reversal
symmetry, that is, the nesting forms of type $A$ and $B$ relative
to $\mathcal{Z}(\boldsymbol{x})$ coincide. The situation for a countable
infinite number of levels, $\alpha\in\mathbb{N}$, is more subtle.
An explanatory example involves the $p$-adic numbers. Let us consider
``positional'' weight $\mu_{\ell}(\boldsymbol{x}):=-\ell\cdot\ln p$
that depend on $\ell\in\mathbb{N}_{0}$ only, and cardinalities $\nu_{\ell}(\boldsymbol{x})\in\{0,1,\dots,p-1\}$.
The corresponding $A$-type partition function 
\begin{eqnarray}
\mathcal{Z}(\boldsymbol{x}) & = & 1\cdot\left(\nu_{0}(\boldsymbol{x})+p^{-1}\cdot\left(\nu_{1}(\boldsymbol{x})+p^{-1}\cdot\left(\nu_{2}(\boldsymbol{x})+p^{-1}\cdot\left(\dots\right.\right.\right.\right. \nonumber \\
 & = & {\displaystyle \sum_{\ell=0}^{+\infty}\nu_{\ell}(\boldsymbol{x})\cdot p^{-\ell}}\label{eq: nested partition function, p-adic A} 
\end{eqnarray}
is the base $p$ expansion of a real number. On the other hand,
$B$-type nesting results in $\mu_{\ell}(\boldsymbol{x})\equiv\ell\cdot\ln p$,
$\ell\in\mathbb{N}_{0}$, and 
\begin{eqnarray}
\mathcal{Z}(\boldsymbol{x}) & = & 1\cdot\left(\nu_{0}(\boldsymbol{x})+p\cdot\left(\nu_{1}(\boldsymbol{x})+p\cdot\left(\nu_{2}(\boldsymbol{x})+p\cdot\left(\dots\right.\right.\right.\right. \nonumber \\
 & = & {\displaystyle \sum_{\ell=0}^{+\infty}\nu_{\ell}(\boldsymbol{x})\cdot p^{\ell}}\label{eq: nested partition function, p-adic B} 
\end{eqnarray}
which is a $p$-adic number. 

Cases with an infinite number of levels can be put into bijection with
other continued nested expansions. Furthermore, the nested structure (\ref{eq: nested partition function}) relies on the linear order $\leq$ on $\mathbb{R}$, so it could be generalized using more general partial orders. These kinds of extensions, and the reduction to the cases of linear order (Proposition \ref{prop: homomorphism totally ordered sets}), deserve more attention in order to better  understand tropicalization methods, their expansions and potential applications. This would go beyond the scope of this work and will be studied in a separate paper.

\section{\label{sec: Tropicalization(s) and the role of ground energy } Tropicalization(s)
and the role of ground energy }

We can now use the filter language to discuss simple models in non-equilibrium
physics. Two elementary remarks on $\mathbb{R}_{\max}=(\mathbb{R}^{\vee},\max,-\infty)$ lead to major implications. Firstly, $\mathbb{R}^{\vee}$ is a grounded poset, and this is linked to the asymmetry between $\max$ and $\min$. Secondly, a presentation of the tropical algebra has a particular symmetry, i.e. idempotence, which will be
relevant in the study of tropicalization processes.

\subsection{\label{subsec: Duality and non-equilibrium} A simple example: duality and non-equilibrium } 

Here, we look at the connections between $\max$\textbackslash{}$\min$
duality, non-equilibrium systems and ground energy in more detail.
To this end, we start from a simple thermodynamic example. Let us
take $N$ microsystems, where each system is defined by a triple $(E_{\alpha},S_{\alpha},T_{\alpha})$,
$\alpha\in[N]$. Here $E_{\alpha}$ is the energy of the microsystem,
$S_{\alpha}=k_{B}\cdot\ln g_{\alpha}(k_{B})$ is its entropy and $T_{\alpha}$
is a ``temperature''. Concrete instances of such a model could describe microscopic ensembles at the moment when they are put in contact, so they are not in equilibrium. However, we do not force the temperature to be Boltzmann
or Gibbs temperature and we do not put constraints on $T_{\alpha}$ in order to work in full generality. 
These quantities define a micro-free energy $F_{\alpha}=F(E_{\alpha},S_{\alpha},T_{\alpha}):=E_{\alpha}-T_{\alpha}\cdot S_{\alpha}$.
One can give the same tropical framework to all the dependent
variables, i.e. micro-free energies, via 
\begin{equation}
\beta:\,\left\{ \frac{F_{\alpha}}{T_{\alpha}},\,\alpha\in[N]\right\} \subseteq\mathbb{R}\hookrightarrow\mathbb{R}^{\wedge}\label{eq: tropicalization micro-free energies}
\end{equation}
where the immersion $\mathbb{R}\hookrightarrow\mathbb{R}^{\wedge}$
acts as the identity on $\mathbb{R}$. So we can think at dependent
variables as real or tropical ones. The function $\beta$ will be
called \emph{tropical temperature} (even if it plays the role of an inverse temperature)  and the $N$ microsystems are
tropically thermalized. Generalizing the results in \cite{AK2015},
the tropical free energy is defined as 
\begin{equation}
F_{\mathrm{trop}}:=\bigoplus_{\alpha\in[N]}\beta\left(\frac{E_{\alpha}-S_{\alpha}\cdot T_{\alpha}}{T_{\alpha}}\right).\label{eq: tropical free energy, B}
\end{equation}
The tropical addition in (\ref{eq: tropical free energy, B}) is identified
with $\min$, so it corresponds to a $B$-type model. The duality
between $A$-type and $B$-type nesting forms can be now restated
in more geometric terms. Indeed, the corresponding $A$-type model
is obtained by means of a projective-like transformation, that are
\begin{equation}
\tilde{E}_{\alpha}:=S_{\alpha},\quad\tilde{S}_{\alpha}:=E_{\alpha},\quad\tilde{T}_{\alpha}:=\frac{1}{T_{\alpha}}.\label{eq: real projective A/B transformation}
\end{equation}
Under previous transformation, (\ref{eq: tropical free energy, B})
becomes 
\begin{eqnarray}
\min_{\alpha\in[N]}\beta\left(\frac{-T_{\alpha}\cdot S_{\alpha}+E_{\alpha}}{T_{\alpha}}\right)&=&\min_{\alpha\in[N]}\beta\left(-\tilde{E}_{\alpha}+\tilde{S}_{\alpha}\cdot\tilde{T}_{\alpha}\right)\nonumber \\
&=&-\max_{\alpha\in[N]}\beta\left(\tilde{E}_{\alpha}-\tilde{S}_{\alpha}\cdot\tilde{T}_{\alpha}\right).\label{eq: tropical free energy, A}
\end{eqnarray}
So one can rewrite the tropical free energy (\ref{eq: tropical free energy, B})
of the $B$-type as an $A$-type expression (\ref{eq: tropical free energy, A}).
Previous transformation corresponds to the conjugation $c\circ\min\circ c=\max$,
where $c:\,\mathbb{R}^{\vee}\longrightarrow\mathbb{R}^{\wedge}$ is
the involution $c(x)=-x$. One can restore the presentation with addition
$\oplus=\min$ via the simultaneous real transformation (\ref{eq: real projective A/B transformation})
and the tropical inversion 
\begin{equation}
\tilde{\beta}=``\frac{1}{\beta}":=-\beta:\,\mathbb{R}\longrightarrow\mathbb{R}^{\vee}.\label{eq: tropical projective A/B transformation}
\end{equation}
That gives 
\begin{equation}
\min_{\alpha\in[N]}\beta\left(\frac{-T_{\alpha}\cdot S_{\alpha}+E_{\alpha}}{T_{\alpha}}\right)=\min_{\alpha\in[N]}\tilde{\beta}\left(\tilde{E}_{\alpha}-\tilde{S}_{\alpha}\cdot\tilde{T}_{\alpha}\right).\label{eq: tropical free energy, A, tropical inversion}
\end{equation}
The inversion ${\displaystyle \tilde{T}_{\alpha}:=\frac{1}{T_{\alpha}}}$
in (\ref{eq: real projective A/B transformation}) is the usual (real)
one for real temperatures $T_{\alpha}$, while ``$\displaystyle \tilde{\beta}:=\frac{1}{\beta}$''
in (\ref{eq: tropical projective A/B transformation}) is the tropical
inversion for tropical temperature $\beta$. It should be
stressed that (\ref{eq: tropical projective A/B transformation})
is not an involution on a set, but rather a homomorphism between two
different grounded monoids, that are $\mathbb{R_{\max}}=(\mathbb{R}^{\vee},\max,-\infty)$
and $\mathbb{R_{\min}}=(\mathbb{R}^{\wedge},\min,+\infty)$. 

So, the different behaviour for ``positive'' and ``negative''
$\beta$ is resolved by a reparametrization. It is worth
remarking that a different approach to the change of behaviour between
positive and negative temperatures was proposed in \cite{Ramsey1956},
where the role of the real parameter ${\displaystyle -\frac{1}{k_{B}T}}$ instead of $T$
was emphasized. On the other hand, the tropical approach and the $\max\backslash\min$
duality in (\ref{eq: tropical projective A/B transformation}) let
one have a common setting where these different physical regimes
are preserved. In the particular case of equilibrium $T_{\alpha}=T$ for a constant
$T$ one recovers the results in \cite{AK2015}. 

The map of a $B$-type non-equilibrium system to an $A$-model can be studied in terms of the choice of the ground energy. In fact, both models are invariant under a shift of free energies $F_{\alpha},\,\alpha \in [N]$. However, the $B$-model has not translation invariance for the energy spectrum $\{E_{\alpha}:\,\alpha \in [N]\}$ in general. This means that, if temperatures $T_{\alpha}$ do not coincide (non-equilibrium case) and 
one reparametrizes $\overline{E}_{\alpha}:=\mathcal{E}+E_{\alpha}$
for any constant $\mathcal{E}\in\mathbb{R}$,
then ${\displaystyle \frac{E_{\alpha}-S_{\alpha}\cdot T_{\alpha}}{T_{\alpha}}<\frac{E_{\gamma}-S_{\gamma}\cdot T_{\gamma}}{T_{\gamma}}}$ is in general not equivalent to ${\displaystyle \frac{\overline{E}_{\alpha}-S_{\alpha}\cdot T_{\alpha}}{T_{\alpha}}<\frac{\overline{E}_{\gamma}-S_{\gamma}\cdot T_{\gamma}}{T_{\gamma}}}$. On the other hand, the $A$-type model
has such a type of symmetry, since $\tilde{E}_{\alpha}-\tilde{S}_{\alpha}\cdot\tilde{T}_{\alpha}<\tilde{E}_{\gamma}-\tilde{S}_{\gamma}\cdot\tilde{T}_{\gamma}$
implies $(\tilde{E}_{\alpha}+\tilde{\mathcal{E}})-\tilde{S}_{\alpha}\cdot\tilde{T}_{\alpha}<(\tilde{E}_{\gamma}+\tilde{\mathcal{E}})-\tilde{S}_{\gamma}\cdot\tilde{T}_{\gamma}$ for all $\tilde{\mathcal{E}}\in\mathbb{R}$. 

A more detailed discussion on this restricted translational invariance of the spectrum and the meaning of the duality $E_{\alpha}\leftrightarrows S_{\alpha}$ would go beyond the scope of this paper and will be discussed elsewhere. 

\subsection{\label{subsec: Tropical
probability and tropical symmetry} Tropical action, normalization and the choice of  ground energy}

The previous example suggests a more general way to
move from real to tropical entities. We will call a \textit{tropicalization}
of a set $\mathcal{R}$ a map 
\begin{equation}
\tau:\,\mathcal{R}\longrightarrow\Lambda\label{eq: tropicalization reals}
\end{equation}
where $\Lambda$ is a tropical semiring. In particular, we are interested in tropicalizations of real variables $\mathcal{R}\subseteq\mathbb{R}^{n}$.
The results in Section \ref{sec: Monoid homomorphisms and a set/element correspondence} give the opportunity to consider simultaneously two objects associated to a monoid. The first is $(\phi(y),\oplus,y)$ in (\ref{eq: principal filters and dual}), $y\in\Lambda$, that is a filter with respect to $\preceq$. The second is $\mathcal{F}_{\Lambda}$ in (\ref{eq: principal filters and dual}).
Proposition \ref{prop: homomorphism and linear order} suggests to
concentrate, firstly, on a totally ordered set in order to include
the presentation $\mathcal{I}_{\Lambda}$. In fact, this is the case
of $\beta\circ F$ in (\ref{eq: tropicalization micro-free energies}),
where $\Lambda=\mathbb{R}^{\wedge}$. So one has at least two additional
tropicalizations on the same poset $\Lambda$, that are 
\begin{eqnarray} 
\iota_{x}:\,\mathcal{R}\longrightarrow\phi(x) & = & \{y\in\Lambda:\,x\preceq y\}, 
\label{eq: tropicalizations reals, global} \\  
\phi:\,\mathcal{R}\longrightarrow\mathcal{F}_{\Lambda} & = & \phi(\Lambda). 
\label{eq: tropicalizations reals, local}
\end{eqnarray}
\label{eq: tropicalizations reals}
The map (\ref{eq: tropicalizations reals, global}) has a physical
relevance in terms of stability. Indeed, a physical system is considered stable if its energy spectrum is bounded from below. This condition is expressed by grounded posets $\phi(x)$, where $x$ plays the role of ground energy and bounds the  elements of $\iota_{x}(\mathcal{R})$ from below. From the statistical point of view, the existence of such $x$ means that the associated tropical probability distribution $W_{n,\mathrm{tr}}=\displaystyle \frac{F_{\mathrm{tr}}-F_n}{T}$ in \cite{AK2015} is normalizable. 

The domains $\Lambda$ and $\mathcal{F}_{\Lambda}$ are homomorphic
as tropical monoids by Proposition \ref{prop: homomorphism and linear order}.
Nevertheless, quite different conclusions can be drawn from these
processes when one looks at the semiring action induced by $\odot=+$.
Indeed, the translation $\varepsilon+\mathcal{R}:=\{\varepsilon+x:\,x\in\mathcal{R}\}$,
$\varepsilon\in\mathbb{R}$, corresponds to actions on $\phi(x)$
and $\mathcal{F}_{\Lambda}$, namely 
\begin{eqnarray} 
\varepsilon\odot y & := & \varepsilon\odot\iota_{x}(y), 
\label{eq: tropical action, global} \\  
\varepsilon\odot\phi(x) & := & \{\varepsilon\odot\tau(y):\,\tau(y)\in\phi(x)\}. 
\label{eq: tropical action, local}
\end{eqnarray}
\label{eq: tropical actions}
The action (\ref{eq: tropical action, global}) is invertible with inverse $``\varepsilon^{-1}"\odot y=(-\varepsilon)\odot\iota_{x}(y)$,
for all $\varepsilon\in\Lambda\backslash\{\infty\}$. So it maps any
filter $\phi(x)$ to another one $\phi(\varepsilon\odot x)$. If one
tropicalizes $\mathcal{R}$ via (\ref{eq: tropicalizations reals, global}),
then the associated tropical action will be called \emph{global}. This means that
(\ref{eq: tropical action, global}) acts simultaneously on the whole
set $\mathcal{R}$ mapping it to another filter $\varepsilon\odot\iota_{x}(\mathcal{R})$.
In particular, the tropicalization $\beta\circ F$ given by the map
(\ref{eq: tropicalization micro-free energies}) is global, with $\Lambda=\mathbb{R}_{\min}$
and $\odot=+$. If the systems are in equilibrium, i.e. $T_{\alpha}=T$
is constant for all $\alpha$, then the global tropical action describes
a different choice of ground energy and a consequent shift of (free)
energies by the same value. So one recovers the invariance of a physical
description under different choices of the ground energy. For systems
out of equilibrium (different $T_{\alpha}$), this tropicalization
does not coincide with the shift of energy levels. 

If instead one uses (\ref{eq: tropicalizations reals, local}) to
tropicalize $\mathcal{R}$, then each element $x\in\mathcal{R}$ is
presented a filter $\phi(x)$ on a certain poset. In the latter approach, the tropical action  will be called \emph{local}, that means that each real variable $x\in\mathcal{R}$ actually represents the choice for the ground value of its image $\phi(x)$. This freedom is important
when one is interested in tropicalizations of probability distributions.
Before moving on to this issue, it should be stressed that some effects
of the invariance of micro-free energies under constant shifts have been
studied in a geometric framework in \cite{AK2016}. More specifically,
Gauss-Kronecker curvature for an ideal statistical mapping vanishes
if and only if there exists a non-vanishing Killing vector field ${\displaystyle \sum_{i=1}^{n}c_{i}\frac{\partial}{\partial x^{i}}}$
for the statistical hypersurface, whose coefficients $c_{i}$ are
constant. This corresponds to a translational symmetry, that is ``global''
(i.e., all $c_{i}$ are equal) in the super-ideal case.

\section{\label{sec: Global and local tropical symmetry}
Global and local tropical symmetry}

The tropicalization (\ref{eq: tropicalization micro-free energies})
of dependent variables $\left\{ F_{\alpha}:\,\alpha\in[N]\right\} $
produces free energies with a tropical symmetry. This means that the
value of the tropical macroscopic free energy (\ref{eq: tropical free energy, B})
does not change if one creates a copy of a certain microsystem $F_{\alpha}$.
The creation of copies affects the counting, and this stresses the role of tropical limit in probability and statistics. 

Both global (\ref{eq: tropicalizations reals, global}) and local
(\ref{eq: tropicalizations reals, local}) tropicalizations induce
tropical symmetry on former real variables. These procedures are connected (Proposition \ref{prop: homomorphism and linear order}) as long as only one variable is involved. However, distinctive features can be extracted from each of these two processes in cases of more variables, depending on tropicalizating them as a whole (the set $\mathcal{R}$) or individually (each element of $\mathcal{R}$ one at a time).  

\subsection{\label{subsec: Global tropical symmetry and statistical amoebas}
Global tropical symmetry and statistical amoebas}

In \cite{AK2015} usual probabilities for events $X\subseteq[N]$
\begin{equation}
W(X)=\frac{\#\left(X\cap\mathfrak{m}_{0}(T)\right)}{\#\mathfrak{m}_{0}(T)}\label{eq: usual probability in =00005B1=00005D}
\end{equation}
were identified, while ``tropical'' probabilities at $k_{B}\ll1$
for states and energy levels are respectively 
\begin{equation}
w_{\alpha,\mathrm{tr}}=-S_{\alpha}+\frac{F_{\mathrm{tr}}(T)-F_{\alpha}(T)}{T}-k_{B}\cdot\ln\left(\#\mathfrak{m}_{0}(T)\right),\quad W_{\alpha,\mathrm{tr}}=w_{\alpha,\mathrm{tr}}+S_{\alpha}.\label{eq: tropical probabilities in =00005B1=00005D}
\end{equation}
At $k_{B}=0$, the weights in (\ref{eq: tropical probabilities in =00005B1=00005D})
are tropically additive and normalized as ${\displaystyle \bigoplus_{\alpha=1}^{N}W_{\alpha,\mathrm{tr}}=0}$.
In particular, one has $\#\mathfrak{m}_{0}(T)=1$ at regular domains
where only one phase $\alpha_{0}\in[N]$ satisfies the minimum for the 
free energy. Here, one gets an ultrafilter probability 
\begin{equation}
W(X)=\left\{ \begin{array}{c}
1,\quad\mbox{if }\alpha_{0}\in X\\
0,\quad\mbox{if }\alpha_{0}\notin X
\end{array}\right..\label{eq: regular domain probability}
\end{equation}

This is a particular instance of a more general application of filters
to probability. In fact, it is well known that a proper filter can
be seen as a $\{0,1\}$-finitely real additive measure on the set
$\Omega$, i.e. a function $\tau:\,\Omega\longrightarrow\{0,1\}$
such that $\tau(\emptyset)=0$ and 
\begin{equation}
{\displaystyle \tau\left(\bigcup_{h=1}^{\ell}\Omega_{h}\right)=\sum_{h=1}^{\ell}\tau(\Omega_{h})}\label{eq: finite additiveness}
\end{equation}
where $\ell\in\mathbb{N}$ and $\Omega_{h}$ are disjoint measurable
subsets of $\Omega$. It is easily shown that the function 
\begin{equation}
\tau(X):=\left\{ \begin{array}{c}
1,\quad\mbox{if }X\in\mathcal{U}\\
0,\quad\mbox{if }X^{c}\in\mathcal{U}\\
\mbox{undefined},\quad\mbox{otherwise}
\end{array}\right.,\quad X\subseteq\Omega\label{eq: filter measure}
\end{equation}
is a $\{0,1\}$-valued measure where countably additivity is relaxed
to finite additivity. Indeed, if $\tau(\Omega_{1})=\tau(\Omega_{2})=1$
then both $\Omega_{1}$ and $\Omega_{2}$ belong to the filter, hence
their intersection is in the filter too. The fact that the empty set
$\emptyset$ does not belong to any proper filter implies that $\Omega_{1}\cap\Omega_{2}\neq\emptyset$.
Thus, there is at most one nonvanishing term in the sum (\ref{eq: finite additiveness}),
where $\Omega_{h}$ are elements of the filter. If $\Omega$ is a
finite set, then $\tau$ in (\ref{eq: filter measure}) is a real
probability measure, since in this case finite additivity is equivalent
to the usual $\sigma$-additivity. 

This interpretation of filters fits well with the tropicalization
(\ref{eq: tropicalizations reals, global}). If $\Omega\longrightarrow\phi(y)$
and $\phi(y)\neq\mathbb{R}^{\vee}$, i.e. $y\neq-\infty$, then any
element $\alpha$ of $\Omega$ can be interpreted as real, since $\Omega\subseteq\mathbb{R}$,
or tropical, as $\iota_{y}(\alpha)$. At the same time, one finds
both real (\ref{eq: usual probability in =00005B1=00005D}) and tropical
weights (\ref{eq: tropical probabilities in =00005B1=00005D}) and
the latter are a result of the global tropicalization (\ref{eq: tropicalization micro-free energies}). 

In this case, the limit process involves $k_{B}$: idempotence for the 
probability $W_{\alpha}=W(\{\alpha\})$ defined in \cite{AK2015}
holds only at the lowest (zeroth) order in $k_{B}$ on the singular locus. This is evident in (\ref{eq: tropical probabilities in =00005B1=00005D}), where the statistical corrections depend on the cardinality $\lambda_{0}(T)=\#\mathfrak{m}_{0}(T)$
defined in (\ref{eq: nesting data, iteration, finite, B}). They correspond to first order corrections in $k_{B}\ll1$ and are the only non-trivial purely perturbative corrections (Proposition \ref{prop: trivial naive perturbation at m>1}).
By the same token, the statistical weights (\ref{eq: usual probability in =00005B1=00005D})
are equal to ${\displaystyle \frac{1}{\lambda_{0}(T)}}$. Thus, idempotence
is lost from the point of view of usual probability weight (\ref{eq: usual probability in =00005B1=00005D}). 

Loss of idempotence is a remarkable phenomenon at the singular locus,
where the free energy is non-differentiable and this can be seen as a
phase transition. In the line of thought that associates phase transitions to a broken symmetry, at a critical temperature $T^{\star}$ the broken tropical symmetry appears as loss of idempotence in statistical prefactors.
The breaking of tropical symmetry on the singular locus is accidental,
i.e. it occurs for certain values of parameters (e.g., temperature).
Furthermore, it is physical in terms of observability by means of
averages of observables with weights (\ref{eq: usual probability in =00005B1=00005D}). In the context of tropical geometry \cite{MS2015}, where one associates a simplicial complex to algebraic tropical functions, the accidental coincidence of phases is described by simplices with non-maximal dimension. 

In general, the addition of a copy of a subsystem makes this broken
symmetry systematic. In fact, one can consider the extension of $\{F_{\alpha}:\,\alpha\in[N]\}$ by a function $F_{N+1}\equiv F_{\alpha_{0}}$ where $F_{\alpha_{0}}(T)\leq F_{\beta}(T)$ for all $\beta\in[N]$ and $T$ in a certain domain. The statistical factor for $W_{\alpha_{0}}(T)$ now involves both $\alpha_{0}\in\mathfrak{m}_{0}(T)$ and $N+1$, i.e.  
\begin{equation}
W_{\alpha_{0}}(T)=\frac{1}{\lambda_{0}(T)}\mapsto\frac{2}{\lambda_{0}(T)+1}\label{eq: usual probability effect copying}
\end{equation}
for any $T$ in the domain. Contrary to the case of phase transitions, averages of observables are unchanged by the addition of a copy in regular domains, so they are not observable in this sense. However, tropical copies can still be identified on the singular locus.  

It is worthy of note that statistical amoebas \cite{AK2016b} provide
one with a geometric formulation for this limit procedure. Indeed,
the instability domain $\mathcal{D}_{k-}$ for a statistical amoeba
(\ref{eq: weight set statistical amoeba}) is induced by the ultrafilter $\mathcal{U}(\alpha(\boldsymbol{x}))$ through $\mathcal{N}_{k-}(\boldsymbol{x})=\mathcal{U}(\alpha(\boldsymbol{x}))\cap\{\mathcal{A}\subseteq[N]:\,\#\mathcal{A}=k\}$, where $\alpha(\boldsymbol{x})$ is the only index in $[N]$ such that $f_{\alpha(\boldsymbol{x})}(\boldsymbol{x})>f_{\beta}(\boldsymbol{x})$
for all $\beta\neq\alpha(\boldsymbol{x})$. So the statistical amoeba
can be used to study singularities of free energy (zeros of (\ref{eq: weight set statistical amoeba})), non-equilibrium domains (where (\ref{eq: weight set statistical amoeba}) is negative) and emergence of a ``macroscopic'' behaviour in domains of maximal instability $\mathcal{D}_{k-}$ (where a filter measure (\ref{eq: filter measure}) is defined). In this context, tropical limits are obtained via the scaling of independent variables, ${\displaystyle x_{i}\mapsto\frac{x_{i}}{k_{B}}}$, or dependent ones ${\displaystyle f_{\alpha}(\boldsymbol{x})\mapsto\frac{f_{\alpha}(\boldsymbol{x})}{k_{B}}}$. 

\subsection{\label{subsec: Local tropical symmetry and the dequantification procedure }
Local tropical symmetry and the dequantification procedure }

In Section \ref{sec: Tropicalization(s) and the role of ground energy } we have pointed out that the local tropicalization (\ref{eq: tropicalizations reals, local}) describes the labeling of a set of systems $\phi(\mathcal{R})$ by their ground energy $x$, with $x\in\mathcal{R}$. We can now explore the statistical effects of individual implementation of idempotence on real variables. Let us take a finite number, say $N$, of distinct microsystems $\Phi:=\left\{ f_{\alpha}(\boldsymbol{x}):\,\alpha\in[N]\right\} $
that defines the statistical model through the partition function
(\ref{eq: partition function}). 

The index space $[N]$ can be immersed in $\tilde{N}:=[N]\times\mathbb{N}$
identifying $\alpha\in[N]$ with $(\alpha,1)\in\tilde{N}$. The limit
procedure can be implemented through the map 
\begin{equation}
\begin{array}{cc}
\begin{array}{c}
\mathrm{T}:\,\tilde{N}\longrightarrow\tilde{N}\\
\mathrm{T}(\alpha,n)=\mathrm{T}_{\alpha}(n):=(\alpha,n+1)
\end{array}, & \alpha\in[N],\,n\in\mathbb{N}\end{array}.\label{eq: tropical copies}
\end{equation}
So $\mathrm{T}_{\alpha}$ describes the addition of a tropical copy
of the microsystem $\alpha$ in the macrosystem, i.e. the disjoint
union 
\begin{equation}
\mathrm{T}_{\alpha}(\Phi):=\left\{ f_{\beta}(\boldsymbol{x}):\,\beta\in[N]\right\} \sqcup\{f_{\alpha}(\boldsymbol{x})\}.\label{eq: tropical copies set}
\end{equation}
One can consider the sets $[N]$, $\Phi$ and $[N]\times\{1\}\subseteq\tilde{N}$
as minimal presentations of the macrosystem since microsystems are
pairwise distinct. If $X\subseteq\tilde{N}$, we will write $\mathrm{T}(X):=X\cup\{\mathrm{T}(x):\,x\in X\}$. 

Now let us identify a ``tropical distribution''
on $[N]$ starting from a standard (real, additive) one $w_{k_{B}}:\,[N]\longrightarrow[0;1]$,
where $k_{B}$ is a parameter that controls the tropicalization process.
From the point of view of standard probability, the creation of a
copy of a dominant microsystem affects statistical weights as in (\ref{eq: usual probability effect copying}).
On the other hand, a tropical ``probability'' should not discern
the addition of copies, since they define the same tropical system. Thus, we assume that the creation of copies does not affect the tropical system. 

This request implies that we can consistently assign tropical weights to a set $X\subseteq\tilde{N}$, starting from a real distribution $w_{k_{B}}$, if $X$ is closed under addition of copies. So, we will say that a set $Y\subseteq\tilde{N}$ is \textit{$\mathrm{T}$-closed} if $\mathrm{T}(Y)=Y$. The \textit{$\mathrm{T}$-closure}
$\overline{X}$ of $X\subseteq\tilde{N}$ is the smallest among all
$\mathrm{T}$-closed sets $Y\subseteq\tilde{N}$ such that $X\subseteq Y$. So a set $Y$ is $\mathrm{T}$-closed if and only if $Y=\overline{Y}$.  Since intersections of $\mathrm{T}$-closed sets are $\mathrm{T}$-closed,
the $\mathrm{T}$-closure of sets is well-defined and its explicit
form is 
\begin{equation}
\overline{X}:=\overline{\mathrm{T}}(X)=\bigcap_{X\subseteq Y=\overline{Y}}Y.\label{eq: replica closure}
\end{equation}

In this setting, $\mathrm{T}$-closed sets represent tropically measurable set. One can get such a tropical measure from weights $w_{k_{B};\alpha}$ assigned to individual copies of the microsystem $\alpha$. In the line of thoughts of \cite{AK2015}, we consider $\mathcal{N}\in\mathbb{N}$ copies of the microsystem $\alpha$ and the tropical limit of $w_{k_{B};\alpha}$ as the simultaneous
limit $k_{B}\rightarrow0$ and $\mathcal{N}\rightarrow\infty$. This
clearly depends on the explicit dependence of weights $w_{k_{B};\alpha}$
from $k_{B}$ and from the relation between $k_{B}$ and $\mathcal{N}$.
For the sake of concreteness, we look at Gibbs weights 
\begin{equation}
w_{k_{B};\alpha}=w_{k_{B};(\alpha,1)}:=\frac{\exp\left(-\frac{f_{\alpha}(T)}{k_{B}}\right)}{{\displaystyle \sum_{\beta\in[N]}}\exp\left(-\frac{f_{\beta}(T)}{k_{B}}\right)}\label{eq: Gibbs' weights}
\end{equation}
and we adopt the prescription ${\displaystyle k_{B}:=\frac{1}{\mathcal{N}}}$. Thus, the role of $k_B$ in this process is to control the creation of tropical copies. 

The addition of $\mathcal{N}-1$ copies
$\alpha\cong(\alpha,1)\mapsto(\alpha,\mathcal{N})$ affects the weight $w_{k_{B};\alpha}$ for the microsystem $\alpha$ as 
\begin{equation}
w_{k_{B};(\alpha,\mathcal{N})}=\frac{\mathcal{N}\cdot\exp\left(-\mathcal{N}\cdot f_{\alpha}(T)\right)}{(\mathcal{N}-1)\cdot\exp\left(-\mathcal{N}\cdot f_{\alpha}(T)\right)+{\displaystyle \sum_{\beta\in[N]}}\exp\left(-\mathcal{N}\cdot f_{\beta}(T)\right)}.\label{eq: Gibbs' weights with copies}
\end{equation}
Now one can consider the limit $k_{B}\rightarrow0^{+}$. If $\alpha\in\mathfrak{m}_{0}(T)$,
then (\ref{eq: Gibbs' weights with copies}) becomes 
\begin{eqnarray}
\displaystyle
w_{0;\alpha} & := & \lim_{\mathcal{N}\rightarrow\infty}\frac{\mathcal{N}\cdot e^{-\mathcal{N}\cdot f_{\alpha}(T)}}{(\mathcal{N}-1)\cdot e^{-\mathcal{N}\cdot f_{\alpha}(T)}+{\displaystyle \sum_{\beta\in\mathfrak{m}_{0}(T)}}e^{-\mathcal{N}\cdot f_{\beta}(T)}+{\displaystyle \sum_{\gamma\notin\mathfrak{m}_{0}(T)}}e^{-\mathcal{N}\cdot f_{\gamma}(T)}} \nonumber \\
& = & \lim_{\mathcal{N}\rightarrow\infty}\frac{\mathcal{N}}{\lambda_{0}(T)-1+\mathcal{N}}=1.\label{eq: tropical, dominant} 
\end{eqnarray}
If instead $\alpha\notin\mathfrak{m}_{0}(T)$, then 
\begin{eqnarray}
\displaystyle
0\leq w_{0;\alpha} & = & \lim_{\mathcal{N}\rightarrow\infty}\frac{\mathcal{N}\cdot e^{-\mathcal{N}\cdot f_{\alpha}(T)}}{(\mathcal{N}-1)\cdot e^{-\mathcal{N}\cdot f_{\alpha}(T)}+{\displaystyle \sum_{\beta\in\mathfrak{m}_{0}(T)}}e^{-\mathcal{N}\cdot f_{\beta}(T)}+{\displaystyle \sum_{\gamma\notin\mathfrak{m}_{0}(T)}}e^{-\mathcal{N}\cdot f_{\gamma}(T)}} \nonumber \\
& \leq & \lim_{\mathcal{N}\rightarrow\infty}\frac{\mathcal{N}}{(\mathcal{N}-1)+{\displaystyle \sum_{\beta\in\mathfrak{m}_{0}(T)}}e^{\mathcal{N}\cdot\left(f_{\alpha}(T)-f_{\beta}(T\right)}}=0 \label{eq: tropical, non-dominant} 
\end{eqnarray}
so $w_{0;\alpha}=0$. These limits rely on both the countable additivity of real probability and the exponential form (\ref{eq: Gibbs' weights}) of Gibbs weights.  

Similarly, one can consider $w_{\mathcal{N}^{-1}}(\mathrm{T}^{\mathcal{N}}(X))$ for $X\subseteq[N]$. Generally, $w_{0}$ is not a real additive distribution. Indeed, for any partition of $\mathfrak{m}_{0}(T)$ in two disjoint sets, say $X_{1}$ and $X_{2}$, $w_{0}(X_{1})+w_{0}(X_{2})\geq 1\geq w_{0}(X_{1}\cup X_{2})=w_{0}(\mathfrak{m}_{0}(T))$. So $w_{0}$ is real additive if and only if, for each possible partition, exactly one set $X_{1}$ or $X_{2}$ is empty. This means that $\#\mathfrak{m}_{0}(T)=1$, so one recovers the ultrafilter probability (\ref{eq: regular domain probability}).

By contrast, $w_{0}$ is tropically additive ($\oplus=\max$) even at $\#\mathfrak{m}_{0}(T)\geq1$.
If $\{X_{n}\}$ is any family of pairwise disjoint subsets of $[N]$,
then $w_{0}{\displaystyle \left(\bigcup_{n}X_{n}\right)=1}$ if and
only if $X_{n}\cap\mathfrak{m}_{0}(T)\neq\emptyset$ for at least
one $n$, that is $\max\{w_{0}(X_{n})\}=1$. So $w_{0}$ is a possibility
distribution, in the sense that they concern the possibility for a certain set of events to happen. 

It is worth remarking that this process involves weights $w_{k_{B};\alpha}$ one at a time, thus one asks for idempotence for each $\alpha$ individually. This corresponds to making copies of the microsystems subsequent to the prior measurement $\alpha\mapsto w_{\alpha}$. If the process was ``global'', then the same number of copies should be created for all the microsystems and usual probabilities (\ref{eq: usual probability in =00005B1=00005D}) would be recovered at $k_{B}\rightarrow 0$. 

Both the real weights (\ref{eq: filter measure}) and the tropical $w_{0}$ take values in $\{0,1\}$. In particular, $\tau$ is a weaker
version of usual probability (\ref{eq: usual probability in =00005B1=00005D}), since it only distinguishes between sure and not sure events. $\{w_{0,\alpha}\}$ is a weaker version of $W_{0,\mathrm{tr}}$ in (\ref{eq: tropical probabilities in =00005B1=00005D}),
since it only provides information on the existence of an element
of $\mathfrak{m}_{0}$ in $X$. In this regard, the procedure used
to derive $w_{0}$ can be called \textit{dequantification}. It should be stressed that the occurrence of a tropical possibility distribution is consistent with the choice of $\mathrm{T}$-closed sets (\ref{eq: replica closure}) as measurable sets.

Also the way in which the dequantification limit is approached is
easily linked to a local tropicalization (\ref{eq: tropicalizations reals, local}). In fact, one can first choose an enumeration for $\mathbb{Q}$,
that is a bijection from $\mathbb{N}$ to $\mathbb{Q}$. Then the copying process (\ref{eq: tropical copies}) moves towards the choice $\Lambda=\mathbb{Q}$
in (\ref{eq: principal filters and dual}). Indeed, once the tropical limit
is reached one has $F_{\alpha}(T)<F_{\beta}(T)$ if and only if $\phi\left(F_{\beta}(T)\right)\subset\phi\left(F_{\alpha}(T)\right)$
and $\iota\left(F_{\alpha}(T)\right)\subset\iota\left(F_{\beta}(T)\right)$,
since the rationals are dense in $\mathbb{R}$.

\section{\label{sec: Conclusions} Conclusions and future perspectives} 

This work was aimed at investigating the links between tropical limit, algebra and statistical physics. 
The above discussion suggests that some physical phenomena can take advantage from a tropical description. A simple algebraic assumption provides a framework where the concepts of dominance, hierarchical distance and composition can be discussed simultaneously. Connections with physical issues can be recognized when one deals with systems that exhibit  ultrametricity, exponential degenerations of energy levels and metastability. 

This opens the way to other questions and proposals. First, it is worth extending the correspondence between elements and subsets looking at other set-theoretic notions. In particular, given a family of sets $\left(\Omega_{n}\right)_{n\in\mathbb{I}}$
indexed by $\mathbb{I}$, one could consider, for any element ${\displaystyle \alpha\in\bigcup_{n}\Omega_{n}}$, a ``dual'' cardinality $\#\alpha$ related to the number of sets $\Omega_{n}$ containing
$\alpha$. If one assumes that each total cardinality ${\displaystyle \sum_{\Omega_{n}:\,\alpha\in\Omega_{n}}\#\alpha=1}$ is 
independent on the number of sets, e.g. ${\displaystyle \#\alpha=\frac{1}{\#\left\{ \Omega_{n}:\,\alpha\in\Omega_{n}\right\} }}$, and the family $(\Phi,\underset{\mathcal{N}\mbox{ times}}{\underbrace{\{\alpha\},\{\alpha\},\dots,\{\alpha\}}})$ for $\mathcal{N}$ copies of the $\alpha$-th microsystem (\ref{eq: tropical copies set}) 
is considered, then $\#\beta=1$ at
$\beta\neq\alpha$ and ${\displaystyle \#\alpha=\frac{1}{\mathcal{N}+1}}$.
The limit ${\displaystyle \frac{1}{\mathcal{N}}\rightarrow0}$ for such a procedure could be formalized in order to understand better the physical meaning behind the $n\rightarrow0$ limit for the dimension of the overlap matrix in replica trick and spin glasses \cite{Parisi1980}, so it deserves a more detailed investigation.

On a broader level, these tools can be useful in the comprehensive
study of different features of complexity. The main advantage pertains to the relation between structural complexity and algebraic rules. The former is the ``hardware'' of a system, e.g. the geometry of a complex networks, and ultrametricity often has a key part in this context. The latter define associative processes, that is the ``software'', and give a basis for extended logics \cite{GuidoToto2008}, including fuzzy logic. So, a tropical micro-macro correspondence and associated tools (e.g., perturbative tropical limit in Section \ref{sec: Perturbative tropical limit})
can help explain connections between the physical structure of complex systems and their underlying logic. This also comes with the dimensionality issue induced by the limit $k_{B}\rightarrow0$ for Boltzmann constant, as already noticed in Sections \ref{sec: Tropical limit and the role of Boltzmann constant} and \ref{subsec: Local tropical symmetry and the dequantification procedure }.
All of this could give new hints on the theoretical framework for
the effectiveness of many methods of statistical physical in current
learning models. 

\section*{Acknowledgements}

I am grateful to Prof. Boris Konopelchenko and Prof. Giulio Landolfi
for their kind comments and continuous support. Part
of this work was written during a research visit at the University
of Loughborough. I would like to thank Prof. Eugene Ferapontov and
Matteo Casati for kind hospitality.

\appendix

\section{\label{sec: Appendix A} Proof of Proposition \ref{prop: filters and ultrametrics}}

Let us first assume that $d(\Omega^{2})$ is grounded. Then, one can check that  
\begin{enumerate}
\item $\mathcal{B}$ is non-empty. Indeed $\Omega\neq\emptyset$, $0\in d(x_{0},\cdot)$ and for all $x_{0}\in\Omega$ the singletons $\{x_{0}\}=\left\{ x\in\Omega:\,d(x,x_{0})\leq0\right\} =S(x_{0},0)$ belong to $\mathcal{B}$.  
\item $\Omega\notin\mathcal{B}$. Indeed, let us assume the contrary and suppose that there exists $x_{0}\in\Omega$ and $r\in d(x_{0},\cdot)$ such that $S(x_{0},r)=\Omega$. Since there is no maximum for $d(\Omega^{2})$,
there exists a pair $(x_{1},x_{2})$ such that $d(x_{1},x_{2})>r$.
But $x_{1},x_{2}\in\Omega=S(x_{0},r)$, then $r<d(x_{1},x_{2})\leq\max\{d(x_{1},x_{0}),d(x_{2},x_{0})\}\leq r$,
contradiction. 
\item The union of any two elements of $\mathcal{B}$ is contained in an element of $\mathcal{B}$. In fact, let us take $S(x_{0},r)$ and $S(y_{0},s)$ with $x_{0},y_{0}\in\Omega$, $r\in d(x_{0},\cdot)$ and $s\in d(y_{0},\cdot)$. So define $M:=\max\{r,s,d(x_{0},y_{0})\}$. Clearly $M\in\left\{ r,s,d(x_{0},y_{0})\right\} \subseteq d(x_{0},\cdot)\cup d(y_{0},\cdot)$ and $S(x_{0},r)\subseteq S(x_{0},M)$, as follows from the definition. Moreover, if $y\in S(y_{0},s)$,  then $d(y,x_{0})\leq\max\{d(y,y_{0}),d(y_{0},x_{0})\}\leq\max\{s,d(x_{0},y_{0})\}\leq M$. Thus, $S(x_{0},r)\cup S(y_{0},s)\subseteq S(x_{0},M)=S(y_{0},M)$, that is $S(x_{0},r)\cup S(y_{0},s)\subseteq S(\bar{x},M)$ with $\bar{x}\in\{x_{0},y_{0}\}$ and $M\in d(\bar{x},\cdot)$. 
\end{enumerate}
So, let us consider the closure of $\mathcal{B}$ under subsets, i.e. $\displaystyle \mathcal{I}:=\left\{ T:\,T\subseteq A,\,A\in\mathcal{B}\right\}$. It satisfies downward closedness by construction and $\Omega\notin\mathcal{I}$ since $\Omega\not\subseteq A$ for all $A\in\mathcal{B}$. Moreover, if $A,B\in\mathcal{I}$, then there exist $x_{A},x_{B}\in\Omega$, $r_{A}\in d(x_{A},\cdot)$ and $r_{B}\in d(x_{B},\cdot)$ such that $A\subseteq S(x_{A},r_{A})$ and $B\subseteq S(x_{B},r_{B})$. By previous observations, there exist $x_{0}\in\Omega$ and $r\in d(x_{0},\cdot)$ such that $A\cup B\subseteq S(x_{A},r_{A})\cup S(x_{B},r_{B})\subseteq S(x_{0},r)$, then $A\cup B\in\mathcal{I}$. This means that $\mathcal{I}$ is an ideal and $\mathcal{F}:=\left\{ \Omega\backslash A:\,A\in\mathcal{I}\right\}$  is a filter by Lemma \ref{lem: ideals/filters power set case}. On
the other hand, if $d(\Omega^{2})$ is not grounded, then there exists
$(\bar{x},\bar{y})\in\Omega^{2}$ such that $d(x,y)\leq d(\bar{x},\bar{y})$
for all $x,y\in\Omega$. This means that $S(\bar{x},d(\bar{x},\bar{y}))=\left\{ x\in\Omega:\,d(x,\bar{x})\leq d(\bar{x},\bar{y})\right\} =\Omega$. So $\Omega\in\mathcal{B}$ and the closure of $\mathcal{B}$ under subsets is the trivial ideal $\mathcal{P}(\Omega)$. 

Now, let $\mathcal{I}$ be an ideal and consider any decreasing positive
function $\mathfrak{d}:\,(\mathcal{I},\subseteq)\longrightarrow(\mathbb{R}_{+}^{\wedge},\leq)$
with $\inf\mathfrak{d}(\mathcal{I})>0$. We will denote with $\mathfrak{D}:\,\Omega^{2}\backslash\{(x,x):\,x\in\Omega\}\longrightarrow\mathcal{P}(\mathbb{R})$
the map 
\begin{equation}
{\displaystyle \mathfrak{D}(x,y):=\left\{ \mathfrak{d}(A):\,A\in\mathcal{I},\,\{x,y\}\subseteq A\right\} .}\label{eq: set of (x,y)-coverings}
\end{equation}
If $x_{1}\neq x_{2}$ then there exist $A_{i}\in\mathcal{I}$ such
that $x_{i}\in A_{i}$, $i\in\{1,2\}$, since ${\displaystyle \Omega=\bigcup_{A\in\mathcal{I}}A}$.
Thus $A_{1}\cup A_{2}\in\mathcal{I}$ from upward directedness of
ideals and $\{x_{1},x_{2}\}\subseteq A_{1}\cup A_{2}$. This means
that $\mathfrak{D}(x_{1},x_{2})\neq\emptyset$, so $d(x_{1},x_{2})\geq\inf\mathfrak{d}(\mathcal{I})>0$.
Moreover, $d$ is symmetric in its entries. Then, for all $x\neq y\neq z\neq x$
in $\Omega$ one finds 
\begin{eqnarray}
 & & \max\{d(x,y),d(y,z)\}\nonumber \\
 & = & \max\left\{ \inf\mathfrak{D}(x,y),\inf\mathfrak{D}(y,z)\right\} \nonumber \\
 & = & \inf\left\{ \max\{\mathfrak{d}(A),\mathfrak{d}(B)\}:\,A,B\in\mathcal{I},\,\{x,y\}\subseteq A,\,\{y,z\}\subseteq B\right\} \label{eq: derivation ultrametric triangle inequality, 1} 
\end{eqnarray}
as follows from the complete distributivity of $\max$ over arbitrary
$\inf$ in the lattice  $\left(\left\{ x\in\mathbb{R}_{+}^{\wedge}:\,x\geq\right.\right.$ $\left.\left.\inf\mathfrak{d}(\mathcal{I})\right\} ,\leq\right)$.
If $A,B\in\mathcal{I}$, then $A\cup B\in\mathcal{I}$, hence downward
closedness of ideals implies $A\cup\{z\}\in\mathcal{I}$ for all $z\in B$.
Thus, decreasing monotony of $\mathfrak{d}$ gives 
\begin{eqnarray}
 & &\inf\left\{ \max\{\mathfrak{d}(A),\mathfrak{d}(B)\}:\,A,B\in\mathcal{I},\,\{x,y\}\subseteq A,\,\{y,z\}\subseteq B\right\} \nonumber \\
& \geq & \inf\left\{ \max\{\mathfrak{d}(A\cup\{z\}),\mathfrak{d}(B\cup\{x\})\}:\,A,B\in\mathcal{I},\,\{x,y\}\subseteq A,\,\{y,z\}\subseteq B\right\} \nonumber \\
& = & \inf\left\{ \mathfrak{d}(A):\,A\in\mathcal{I},\,\{x,y,z\}\subseteq A\right\} \nonumber \\
& \geq & \inf\left\{ \mathfrak{d}(A):\,A\in\mathcal{I},\,\{x,z\}\subseteq A\right\} =d(x,z).\label{eq: derivation ultrametric triangle inequality, 2} 
\end{eqnarray}
where last inequality comes from $\{A\in\mathcal{I},\,\{x,y,z\}\subseteq A\}\subseteq\{A\in\mathcal{I},\,\{x,z\}\subseteq A\}.$
Indeed, it is an equality because of decreasing monotony of $\mathfrak{d}$.
Finally, with a slight abuse of notation, one can denote by the same
symbol the extension of $d$ to $\Omega^{2}$ such that $d(x,x)=0$,
$x\in\Omega$. Thus $d$ is symmetric in its arguments, vanishes if
$x=y$, is positive if $x\neq y$, and verifies the ultrametric triangle
inequality. So, $d$ is an ultrametric.

\section{\label{sec: Appendix B} Proof of Proposition \ref{prop: homomorphism totally ordered sets}}

First, it is worth pointing out the following observation. Let $(\Lambda,\preceq)$ be a join-complete semilattice. If  $z_{1}:=\sup\{\sup X,\sup Y\}$, then $x\preceq\sup X\preceq z_{1}$
for all $x\in X$ and $y\preceq\sup Y\preceq z_{1}$ for all $y\in Y$.
Thus $u\preceq z_{1}$ for all $u\in X\cup Y$ and $z_{2}:=\sup(X\cup Y)\preceq z_{1}$.
On the other hand, from $\sup X\preceq\sup(X\cup Y)=z_{2}$ and similarly
$\sup Y\preceq z_{2}$ one has $z_{1}=\sup\{\sup X,\sup Y\}\preceq z_{2}$.
This means that 
\begin{equation}
\sup\{\sup X,\sup Y\}=\sup\{X\cup Y\}.\label{eq: sup sups =00003D sup join}
\end{equation}  

So, let us move to the proof of the proposition. Let $\Delta$ be a totally ordered set. 
\begin{enumerate}
\item Let us assume that $\psi:\,(\Delta,\max,-\infty)\longrightarrow(\Lambda,\oplus,\bot)$
is a monoid homomorphism and take any $\psi(a),\psi(b)\in\psi(\Delta)$.
From $\max\{a,b\}\in\{a,b\}$ one gets $\psi(a)\oplus\psi(b)=\psi(\max\{a,b\})\in\{\psi(a),\psi(b)\}$,
then $\psi(\Delta)$ is totally ordered. 
\item Now let $\phi_{\vartheta}$ in (\ref{eq: ideals for filter presentation})
be a monoid homomorphism, so $\vartheta(\max\{a,b\})\in\phi_{\vartheta}(\max\{a,b\})=\phi_{\vartheta}(a)\cup\phi_{\vartheta}(b)$.
Hence $\vartheta(\max\{a,b\})\preceq\sup\{\vartheta(a),\vartheta(b)\}$.
Moreover, if $u,v\in\Delta$ and $\max\{u,v\}=v$, then $\phi_{\vartheta}(u)\subseteq\phi_{\vartheta}(v)$.
So $\vartheta(u)\in\phi_{\vartheta}(u)\subseteq\phi_{\vartheta}(v)$
means that $\vartheta(u)\preceq\vartheta(v)$ and $\vartheta$ is
increasing. Thus, $\vartheta(a)\preceq\vartheta(\max\{a,b\})$ and
$\vartheta(b)\preceq\vartheta(\max\{a,b\})$, i.e. $\sup\{\vartheta(a),\vartheta(b)\}\preceq\vartheta(\max\{a,b\})$.
One finally gets $\sup\{\vartheta(a),\vartheta(b)\}=\vartheta(\max\{a,b\})$,
that means $\vartheta(a)\oplus\vartheta(b)=\vartheta(\max\{a,b\})$
by (\ref{eq: sum sup}). Furthermore, from $\vartheta(-\infty)\in\phi_{\vartheta}(-\infty)$
and the homomorphism condition $\phi_{\vartheta}(-\infty)=\{\bot\}$
one deduces $\vartheta(-\infty)=\bot$. So the mapping $\vartheta$
is a monoid homomorphism and the poset $\vartheta(\mathbb{R}^{\vee})$
is a totally ordered set in $\Lambda$. 
\item Let us introduce $\hat{\psi}:=\mathrm{id}\circ\psi:\,\Delta\longrightarrow\mathcal{P}(\hat{\Lambda})$,
where $\hat{\Lambda}$ is defined in (\ref{eq: almost complete, extension})
and $\mathrm{id}:\,\mathcal{P}(\Lambda)\hookrightarrow\mathcal{P}(\hat{\Lambda})$
is the immersion $S\subseteq\Lambda\mapsto S\subseteq\hat{\Lambda}$.
The mapping $\hat{\vartheta}(a):=\sup\hat{\psi}(a)$, $a\in\Delta$,
satisfies $\hat{\vartheta}(-\infty)=\sup\hat{\psi}(-\infty)=\sup\{\bot\}=\bot$
and 
\begin{eqnarray}
\hat{\vartheta}(a)\oplus\hat{\vartheta}(b)&=&\sup\{\hat{\vartheta}(a),\hat{\vartheta}(b)\}=\sup(\hat{\psi}(a)\cup\hat{\psi}(b))\nonumber \\
&=&\sup\hat{\psi}(\max\{a,b\}).\label{eq: compatibility reduced monoid homomorphism}
\end{eqnarray}
Previous equalities come from (\ref{eq: sum sup}), (\ref{eq: sup sups =00003D sup join}),
the definition of $\hat{\vartheta}$ and the assumption that $\psi$
is a monoid homomorphism, which implies that $\hat{\psi}$ is a
monoid homomorphism too. Hence $\hat{\vartheta}:\,\Delta\longrightarrow\hat{\Lambda}$
is a monoid homomorphism. Let $\Delta_{0}:=\hat{\vartheta}^{(-1)}(\Lambda)$,
that is a totally ordered subset of $\Delta$ with $-\infty \in \hat{\vartheta}^{(-1)}(\bot)\subseteq\Delta_{0}$. One can now define
the restriction $\vartheta:=\hat{\vartheta}|_{\Delta_{0}}$ and the
map $\bar{\iota}:=\iota\circ\vartheta:\,\Delta_{0}\overset{\vartheta}{\longrightarrow}\Lambda\overset{\iota}{\longrightarrow}\mathcal{P}(\Lambda)$.
From part (1.), the set $\vartheta(\Delta_{0})$ is totally ordered,
so the restriction of $\iota$ to $\vartheta(\Delta_{0})$ is a monoid
homomorphism as follows from Proposition \ref{prop: homomorphism and linear order}.
Thus $\bar{\iota}$ is a monoid homomorphism too and $\psi(a)\subseteq\bar{\iota}(a)\cup\{\vartheta(a)\}=\phi_{\vartheta}(a)$. 
\end{enumerate}

\section{\label{sec: Appendix C} Proof of Proposition \ref{prop: trivial naive perturbation at m>1}}

Let us write $k:=k_{B}$ and $\displaystyle \partial_{k}:=\frac{d}{dk}$ for notational convenience and make explicit reference to the temperature $T$ introducing $\displaystyle F_{\alpha}:=T\cdot f_{\alpha}$.
One finds that tropical free energy corresponds to the $0$-th order
term in the expansion ${\displaystyle \frac{F}{T}\Big|_{k=0}=\min\left\{ \frac{F_{\alpha}}{T}:\,\alpha\in[N]\right\} =\kappa_{0}}$.
First order correction corresponds to statistical prefactors in \cite{AK2015}
and they coincide with results from standard perturbation theory.
In fact, one has 
\begin{eqnarray}
\left(\partial_{k}\frac{F}{T}\right)\Big|_{k=0} & = & \left[\partial_{k}(-k\ln\mathcal{Z})\right]|_{k=0}\nonumber \\
& = & -\ln\left(\lambda_{0}+{\displaystyle \sum_{\alpha\in\mathfrak{s}_{0}}\exp}{\displaystyle \frac{T\kappa_{0}-F_{\alpha}}{kT}}\right)\Big|_{k=0}-k\cdot\partial_{k}\ln\mathcal{Z}|_{k=0} \nonumber \\
& = & -\ln\lambda_{0}-\left.\frac{\sum_{\alpha\in\mathfrak{s}_{0}}\frac{F_{\alpha}-T\kappa_{0}}{kT}\cdot\exp\frac{T\kappa_{0}-F_{\alpha}}{kT}}{\lambda_{0}+\sum_{\alpha\in\mathfrak{s}_{0}}\exp\frac{T\kappa_{0}-F_{\alpha}}{kT}}\right|_{k=0}=-\ln\lambda_{0}.\label{eq: first order correction tropical} 
\end{eqnarray}
Now let us consider higher order contributes. One has ${\displaystyle \partial_{k}^{m}\frac{F(k)}{T}=\partial_{k}^{m}\left(\frac{F(k)}{T}-\kappa_{0}\right)}$,
so we can take ${\displaystyle \kappa_{0}=\frac{F_{\mathrm{trop}}}{T}\equiv0}$
without loss of generality. In particular, this means that ${\displaystyle \lim_{k\rightarrow0^{+}}\mathcal{Z}(k)=\lambda_{0}\neq0}$.
Thus, one has 
\begin{equation}
\left.\partial_{k}^{m}\mathcal{Z}\right|_{k=0}=\left.\sum_{\alpha\in\mathfrak{s}_{0}}\partial_{k}^{m}\exp\frac{-F_{\alpha}}{kT}\right|_{k=0}=\left.\sum_{\alpha\in\mathfrak{s}_{0}}Q_{m}\left(k^{-1};-\frac{F_{\alpha}}{T}\right)\cdot\exp{\displaystyle \frac{-F_{\alpha}}{kT}}\right|_{k=0}=0\label{eq: perturbation Z m-order}
\end{equation}
where ${\displaystyle Q_{m}\left(k^{-1};-\frac{F_{\alpha}}{T}\right)}$
is a polynomial in $k^{-1}$ with coefficients that depend on $-{\displaystyle \frac{F_{\alpha}}{T}}$.
Moreover ${\displaystyle \partial_{k}\left(\frac{1}{\mathcal{Z}}\right)=-\frac{\partial_{k}\mathcal{Z}}{\mathcal{Z}^{2}}\underset{{\scriptstyle k\rightarrow0}}{\longrightarrow}0}$.
Now assume that $\partial_{k}^{l}{\displaystyle \frac{1}{\mathcal{Z}}=0}$
for all $1\leq l\leq m-1$. Then 
\begin{equation}
0=\partial_{k}^{m} \left(\frac{1}{\mathcal{Z}}\cdot\mathcal{Z}\right)=\sum_{l=0}^{m}{{m}\choose{l}}\partial_{k}^{l}\frac{1}{\mathcal{Z}}\cdot\partial_{k}^{m-l}\mathcal{Z}.
\label{eq: eq: perturbation Z m-order, 2} 
\end{equation}
From the inductive hypothesis one gets ${\displaystyle 0=\frac{1}{\mathcal{Z}}\cdot\partial_{k}^{m}\mathcal{Z}+\mathcal{Z}\cdot\partial_{k}^{m}\frac{1}{\mathcal{Z}}}$.
But $\partial_{k}^{m}\mathcal{Z}=0$ at $m\geq1$ from (\ref{eq: perturbation Z m-order}) and ${\displaystyle \lim_{k\rightarrow0^{+}}\mathcal{Z}(k)\neq0}$.
Hence ${\displaystyle \partial_{k}^{m}\frac{1}{\mathcal{Z}}=0}$ for
all $m\geq1$ by induction. So $\partial_{k}^{m-l}\mathcal{Z}|_{k=0}=0$ if $m>l$, ${\displaystyle \partial_{k}^{l}\frac{1}{\mathcal{Z}}}$
is vanishing if $l>0$ and finite at $l=0$. Given that 
\begin{equation}
\partial_{k}^{m}\ln\mathcal{Z}=\partial_{k}^{m-1}\frac{\partial_{k}\mathcal{Z}}{\mathcal{Z}}=\sum_{l=0}^{m-1}{{m-1}\choose{l}}\partial_{k}^{l}\frac{1}{\mathcal{Z}}\cdot\partial_{k}^{m-l}\mathcal{Z},
\label{eq: eq: perturbation Z m-order, 3}
\end{equation}
this means that $\partial_{k}^{m}\ln\mathcal{Z}=0$
for all $m\geq1$. Thus
\begin{eqnarray}
\partial_{k}^{m}\frac{F}{T} & = & \partial_{k}^{m}(-k\ln\mathcal{Z})=-\sum_{l=0}^{m}{{m}\choose{l}}\cdot\partial_{k}^{l}k\cdot\partial_{k}^{m-l}\ln\mathcal{Z}\nonumber \\ 
& = & -k\cdot\partial_{k}^{m}\ln\mathcal{Z}-m\cdot\partial_{k}^{m-1}\ln\mathcal{Z}\label{eq: m-th k-derivative F/T}
\end{eqnarray}
which vanishes at $k_{B}\rightarrow0^{+}$ and $m>1$. 



\end{document}